\documentclass[10pt]{amsart}
\usepackage[utf8]{inputenc}
\usepackage{amsmath}
\usepackage{amssymb}
\usepackage{graphicx}
\usepackage{algorithm,algorithmic}
\usepackage{appendix}
\usepackage{dsfont}

\newtheorem{de}{Definition}[section]
\newtheorem{theo}{Theorem}[section]
\newtheorem{prop}{Proposition}[section]
\newtheorem{lemme}{Lemma}[section]
\newtheorem{cor}{Corollary}[section]
\newtheorem{remark}{Remark}[section]
\renewcommand{\bigwedge}{\Lambda}
\newcommand{\Diff}{\text{Diff}}
\newcommand{\Cte}{\text{Cte}\,}
\DeclareMathOperator{\mdiv}{div}

\title[Varifolds]{The varifold representation of non-oriented shapes for diffeomorphic registration}
\author{Nicolas Charon, Alain Trouv\'{e}}

\begin{document}

\maketitle

\begin{abstract}
In this paper, we address the problem of orientation that naturally arises when representing shapes like curves or surfaces as currents. In the field of computational anatomy, the framework of currents has indeed proved very efficient to model a wide variety of shapes. However, in such approaches, orientation of shapes is a fundamental issue that can lead to several drawbacks in treating certain kind of datasets. More specifically, problems occur with structures like acute pikes because of canceling effects of currents or with data that consists in many disconnected pieces like fiber bundles for which currents require a consistent orientation of all pieces. As a promising alternative to currents, varifolds, introduced in the context of geometric measure theory by F. Almgren, allow the representation of any non-oriented manifold (more generally any non-oriented rectifiable set). In particular, we explain how varifolds can encode numerically non-oriented objects both from the discrete and continuous point of view. We show various ways to build a Hilbert space structure on the set of varifolds based on the theory of reproducing kernels. We show that, unlike the currents' setting, these metrics are consistent with shape volume (theorem \ref{theo:principal}) and we derive a formula for the variation of metric with respect to the shape (theorem \ref{theo:variation_formula}). Finally, we propose a generalization to non-oriented shapes of registration algorithms in the context of Large Deformations Metric Mapping (LDDMM), which we detail with a few examples in the last part of the paper.

\end{abstract}

\section{Introduction}
\label{sec:intro}
Statistical shape analysis, the study of shape variability among a group of subjects, has been an active research field since the early works of \cite{Dryden1998}, \cite{Grenander1998}  and \cite{Kendall1999}. Important applications have been found notably in medical imaging leading to the growing community of computational anatomy. The general approach of computational anatomy is to consider shape variability between subjects as resulting from deformations of the ambient space that one tries to estimate. Thus, a fundamental requirement for most algorithms is to have a relevant measure of residual shape difference. One of the main feature in this domain, though, is the wide variety of different shape structures coming from data measurements~: images, landmarks, unlabeled distributions of points, curves, curve bundles, surfaces... These last examples have been gathering many efforts, mostly due to the difficulty to design satisfying data attachment distances. In that regard, a particularly elegant framework has been proposed in \cite{Glaunes}, introducing in computational anatomy the notion of \emph{mathematical currents}, which we briefly present in section \ref{sec:cur}. The construction of Hilbert metrics on currents' spaces has lead to robust distances between shapes that are also very convenient from a numerical point of view, which has proved of high interest in terms of application to shape registration between curves and surfaces (\cite{Glaunes2},\cite{Durrlemann4}).\\

Despite such important successes in computational anatomy, several critics have recently emerged related the use of currents for shape analysis. In our point of view, most of them are the consequence of the orientation dilemma that is totally inherent to the modelling of shapes as currents. We have gathered a synthesis of those problems in section \ref{sec:cur}. The central question would be to provide a representation that is completely free of shape orientation. This issue, however, is not new in the field of geometric measure theory from which the notion of currents originally appeared. Alternative  concepts have been introduced and studied from a theoretical angle, notably the idea of \textit{varifolds} suggested by F.Almgren and later developed by W.Allard. The purpose of this paper is precisely to adapt varifolds to the problematics of computational anatomy, which are fairly different from geometric measure theory. This is not the only possible approach to mention though, since the concept of \textit{normal cycles} has also drawn recent attention in shape analysis \cite{Fu},\cite{Cohen-Steiner}. \\
In the third section of the paper, we first present the mathematical object of varifolds as measures on the Grassmannian bundle, in the spirit of \cite{Allard}. The real new challenge compared to geometric measure theory is to define metrics between varifolds that are simultaneously 'fitted' to the comparison of non-oriented shapes and easily computable from a practical point of view. As for currents, a convenient way is given by the reproducing kernels' setting. Section \ref{sec:var_kernel} is therefore focused on the definition of kernels for the space of varifolds and shows the fundamental result of theorem \ref{theo:principal}. We also give a theoretical formula of variation of such metrics with respect to the geometrical support of shapes in theorem \ref{theo:variation_formula}. Finally, we present a first set of applications of our framework with an extension of classical large deformation matching algorithm to unoriented curves and surfaces.

\section{Representation of shapes with currents~: strengths and limitations}
\label{sec:cur}

\subsection{Currents in computational anatomy}
In the following, we shall briefly present the approach of currents in computational anatomy as first introduced in this field in \cite{Glaunes} and developed later on in \cite{Durrleman}. The essential features of currents, as we shall see, is that it provides an embedding of the set of all shapes of given dimension into a common Banach space and gives an \textbf{intrinsic} representation, in the sense that it is independent of shape parametrization. \\
Throughout the text, we will adopt the following notations and definitions~:
\begin{itemize}
 \item $E$ : the embedding $n$-dimensional euclidean space of shapes.
 \item $\bigwedge^{p} E$ ($0 \leq p \leq n$) : $p$-times exterior product of $E$, which is a vector space of dimension $\binom{n}{p}$ spanned by the set of simple $p$-vectors $\xi_{1}\wedge..\wedge\xi_{p}$.
 \item $\bigwedge^{p} E$ is equipped with the usual euclidean metric given for two simple $p$-vectors $\xi=\xi_{1}\wedge..\wedge\xi_{p}$ and $\eta=\eta_{1}\wedge..\wedge\eta_{p}$ by the determinant of the Gram matrix $\langle \xi, \eta \rangle= \text{det}(\langle \xi_{i}, \eta_{j}\rangle)_{i,j}$. In particular, $|\xi|$ gives the volume of the corresponding parallelotope.
 \item $\Omega_{0}^{p}(E):=C_{0}^{0}(E,\bigwedge^{p} E^{\ast})$~: the set of continuous $p$-dimensional differential forms on $E$ vanishing at infinity. $\Omega_{0}^{p}(E)$ equipped with the infinite norm is thus a Banach space.
\end{itemize}
This leads to the following definition of $p$-currents on $E$~:
\begin{de}
 $\Omega_{0}^{p}(E)'$, the space of all continuous linear forms on the space of differential forms, is by definition the space of $p$-currents on $E$.
\end{de} 
\noindent Note that in the particular case $p=0$, the previous definition is exactly the one of usual distributions on $E$ (dual of the space of real-valued functions). Just as for distributions, simplest examples of currents are given by \textbf{Dirac currents} $\delta_{x}^{\xi}$ with $x \in E, \ \xi \in \bigwedge^{p} E$ such that for any differential for $\omega \in \Omega_{0}^{p}(E)$, we have $\delta_{x}^{\xi}(\omega)=\omega_{x}(\xi)$. \\
Now, the relationship between shapes and currents lies fundamentally in the fact that \textit{every d-dimensional oriented sub-manifold $X$ of $E$ of finite volume can be represented by an element of $\Omega_{0}^{d}(E)'$}. It is indeed a classical result from integration theory (\cite{Federer},\cite{Lang}) that any $d$-dimensional differential form can be integrated along $X$ and thus~:
\begin{equation}
\label{eq:def_rect_current}
 C_{X}: \ \omega \mapsto \int_{X} \omega 
\end{equation}
defines an element of $\Omega_{0}^{p}(E)'$. The space of $d$-currents is much larger though and contains many other interesting objects, among them \textit{rectifiable subsets} of $E$, which are the generalized submanifolds in the point of view of measure theory, can be treated as currents in the exact same way as submanifolds, by integrating differential forms~: such currents are called \textit{rectifiable currents} in the literature (cf \cite{Federer} and \cite{Morgan}). \\
Another major interest of the previous is the consistency of the currents representation between the continuous and the discrete setting. Any discrete shape given as a set of points with a mesh can be transcribed into a current by associating to each simplex of the mesh a Dirac current located at the center with a simple $d$-vector encoding the local volume element. It can be shown that this discrete representation converges to the continuous one given by equation \ref{eq:def_rect_current} as the mesh becomes finer~: we refer to \cite{Durrleman} for more details. In numerical terms, this means that a shape can be encoded as a finite collection of point positions and $d$-vectors. \\
The last important element to address is the construction of a metric on the space of currents that can be both computed explicitly (at least for discrete surfaces) and that transcribes as well as possible the intuitive idea of closeness between two shapes. A particularly nice framework was proposed in \cite{Glaunes} by joining currents with the theory of \textbf{reproducing kernel Hilbert space} (RKHS). This approach consists in defining a vector kernel on E ($K:E \times E \rightarrow \mathcal{L}(\bigwedge^{p}E)$) and its associated RKHS $W$. It is then shown that, under some regularity assumptions on the kernel, the space of $p$-currents is continuously embedded in the dual $W'$ which is also a Hilbert space. As a result, currents and therefore shapes can be compared using a Hilbert norm which has the additional property that the inner product between two Dirac currents writes~:
\begin{equation*}
 \langle \delta_{x_{1}}^{\xi_{1}}, \delta_{x_{2}}^{\xi_{2}}\rangle_{W'}=\langle \xi_{1}, K(x_{1},x_{2})\xi_{2} \rangle
\end{equation*}
allowing explicit computations when considering discrete shapes represented as finite sums of diracs. Such norms can be also proved to have several interesting theoretical properties. For instance, kernel norms of a rectifiable current is dominated by its Hausdorff measure (or its $d$-dimensional volume)~:
\begin{prop}
\label{norm_current_Hausdorff}
 If $C_{X}$ is a rectifiable current associated to a $d-$dimensional rectifiable set $X$, then~:
\begin{equation*}
 \|C_{X}\|_{W'} \leq \Cte \mathcal{H}^{d}(X)
\end{equation*}
where $\Cte$ is a constant depending only the choice of the kernel.  
\end{prop}
\begin{proof}
For all $\omega \in W$, we have~:
\begin{eqnarray*}
 |C_{X}(\omega)| &\leq& \int_{X} \|\omega\|_{\infty} \\
 &\leq& |\omega|_{\infty} \mathcal{H}^{d}(X) \\
 &\leq& \Cte \|\omega\|_{W} \mathcal{H}^{d}(X)
\end{eqnarray*}
where the last inequality results from the continuous embedding of $\Omega_{0}^{d}(E)$ into $W$. Thus, $\|C_{X}\|_{W'} \leq \Cte \mathcal{H}^{d}(X)$.
\end{proof}
\noindent However, due to orientation, the converse domination does not hold, as we shall see in the following. \\
Such kernel metrics on currents have been successfully used to derive data attachment distances, allowing an extension of Large Deformation Diffeomorphic Metric Mapping (LDDMM) algorithms for registration of curves or surfaces (\cite{Glaunes},\cite{Glaunes2},\cite{Durrlemann4}) as well as statistical estimation of templates. Nevertheless, currents are intrinsically modelling oriented objects~: we shall show in the following subsection the limitations that this induces in terms of classical applications in shape analysis.

\subsection{Currents and the issue of shape orientation}
As mentioned previously, the current associated to a given $d-$dimensional rectifiable subset of $E$ depends on the orientation of $X$. If for instance $X$ is connected and $\check{X}$ denotes the same set but with opposite orientation then, in the space of currents, $C_{\check{X}}=-C_{X}$. It results that, in order to compare two shapes in the currents' setting, it is absolutely compulsory to have orientations of both shapes which must be in addition consistent with each other. A second point is that currents with opposite orientation within a small space domain may \textbf{cancel} each other with respect to the kernel metric. Indeed, for the case of a Gaussian kernel with scale $\sigma$ ~:  
\begin{equation}
\label{currents_cancelling}
 \|\delta_{x}^{\xi}-\delta_{y}^{-\xi}\|_{W'} = 2|\xi|^2 \left( 1-e^{-|x-y|^2/\sigma^2} \right)
\end{equation}
which vanishes whenever $|x-y|$ is small compared to $\sigma$. This trivial fact has important consequences in practice. \\
Difficulties can first occur with non-orientable shapes on which integration of differential forms like in equation (\ref{eq:def_rect_current}) is not well-defined and thus currents cannot handle unorientable shapes in a satisfying way. This first drawback remains anecdotal though, since most shapes in computational anatomy are orientable. However, even if orientable, orientating shapes \textit{consistently} can be either a difficult or even ill-posed problem in certain datasets. This is notably the case for fiber bundles in the 2D or 3D space consisting of many different and possibly disconnected pieces of curves. Indeed, if a given shape has $N$ connected components, there are $2^{N}$ different possible orientations. In the example of figure \ref{matching_misoriented_fibres}, we have shown what could occur in terms of matching if bundles of small curves are oriented randomly. Using currents with such objects requires a way to propagate orientation from one part of the shape to all the rest. Technically, this leads to additional pre-processing which can be particularly tricky for highly disconnected bundle of curves with many different directions like the ones corresponding to white matter fibers estimated from DTI in brain imaging, we show the example of one subject in figure \ref{brain_fibers}. This example also emphasizes another difficulty appearing with sets of curves crossing each other for which the very notion of 'consistent' orientation may become meaningless. In such cases, the orientation is clearly irrelevant and one would like to treat objects as sets of unoriented shapes, which is impossible with currents.

\begin{figure}
\leftskip -1cm
\begin{tabular}{cc}
\includegraphics[width=7cm,height=7cm]{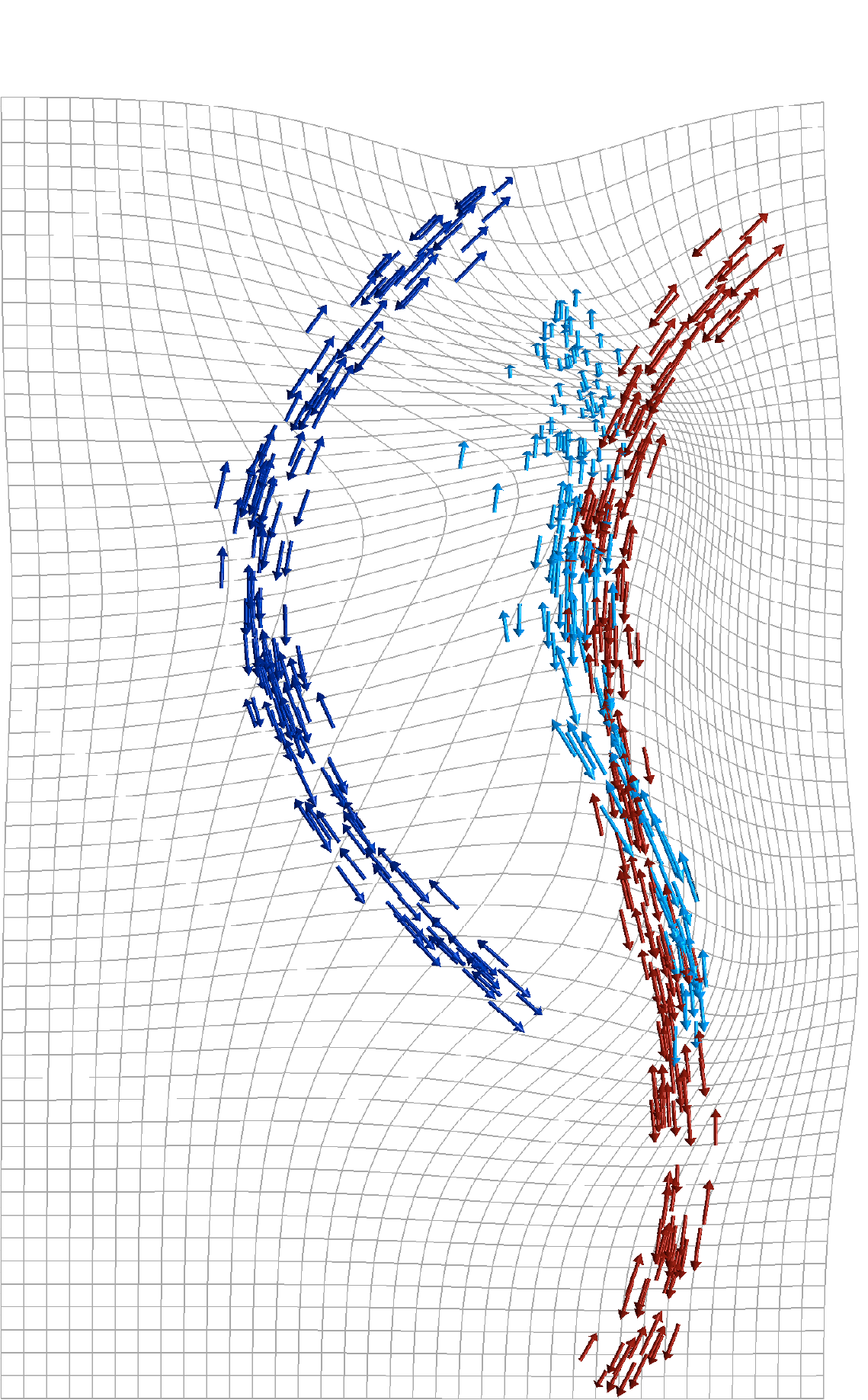} & \includegraphics[width=7cm,height=7cm]{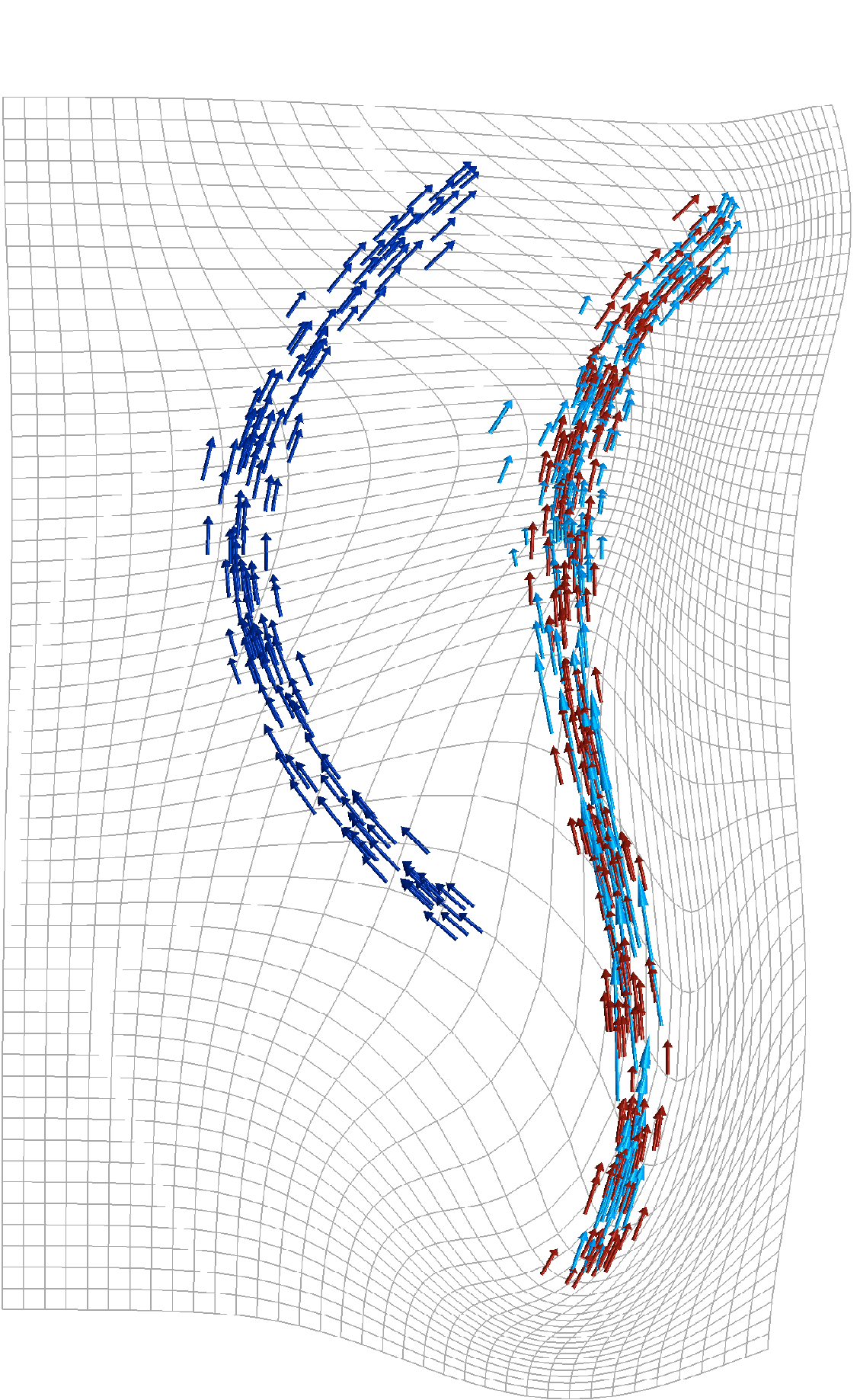}
\end{tabular}
\caption{Example of a matching between two fiber bundles consisting of many different pieces of curves, each of them being oriented randomly (left figure) versus a consistent orientation of all pieces (right figure). The source shape is in dark blue, the target in red and the deformed source is in light blue. We see that a random orientation does not provide a satisfying matching result}
\label{matching_misoriented_fibres}
\end{figure}

\begin{figure}
\includegraphics[width=6cm]{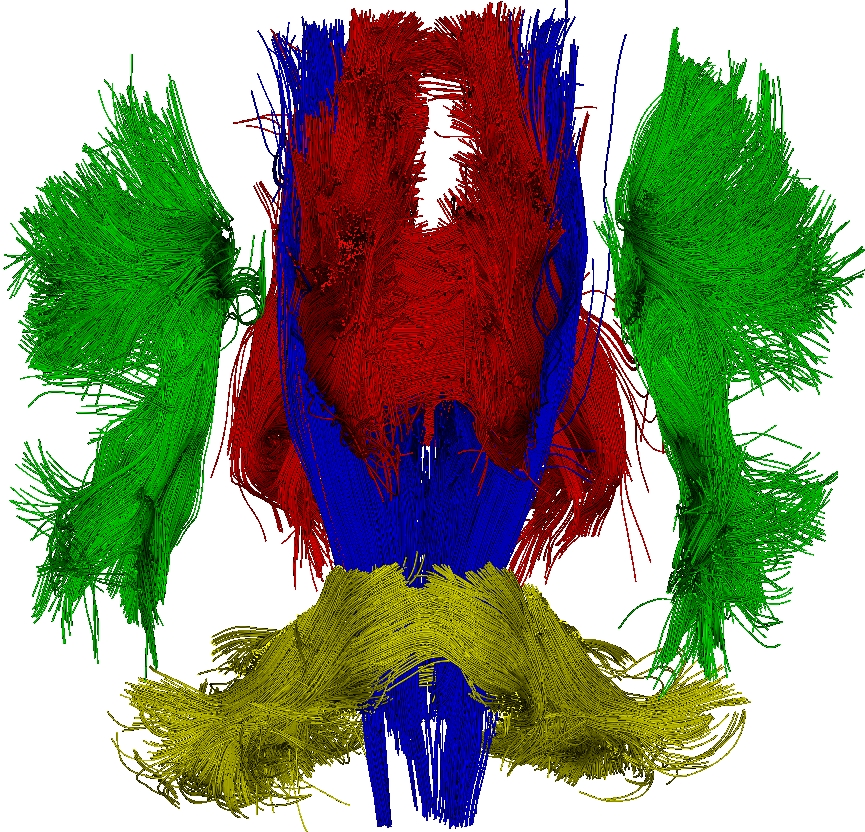}
\caption{An example of white matter fiber bundle estimated from Diffusion Tensor Imaging (DTI) illustrating the potential difficulty of consistent orientation of all different fibers.}\label{brain_fibers}
\end{figure} 

Finally, orientation may represent an obstacle even in the simplest case of usual oriented and connected submanifolds because  structures like sharp spines or tails naturally lead, when represented in a kernel Hilbert space of currents, to annihilation of some parts of the shapes. This comes again from equation (\ref{currents_cancelling}), or, to state it in a more general way, from the fact that, while the $W'$-norm of a set is always dominated by its Hausdorff measure as we saw in proposition \ref{norm_current_Hausdorff}, there exists sets of given Hausdorff measure but with arbitrarily small $W'$-norm. The plane curve of figure \ref{CEX_currents} is an example. Indeed, the $1$-Hausdorff measure being the length of the curve~: $\mathcal{H}^{1}(\gamma)=2(1+\epsilon)\geq 2$. On the other hand, if $\gamma$ is parametrized on an interval $I$, then we have~: 
\begin{equation*}
 \|C_{\gamma}\|_{W'}^{2} = \int_{I \times I} \gamma'(s)^{T}K(\gamma(s),\gamma(t))\gamma(t) ds dt
\end{equation*} 
Let's assume, to simplify computations, that kernel $K$ is a translation invariant kernel of the form $K(x,y)=k(x-y)\,\text{Id}_{\mathbb{R}^2}$ with Lipschitz regularity. Then, by expanding the integrals as the sum of the four pieces of curves $\gamma_{v_{1}}$, $\gamma_{v_{2}}$, $\gamma_{h_{1}}$, $\gamma_{h_{2}}$, we can eventually simplify the result thanks to the translation-invariance, which gives~:
\begin{equation*}
 \|C_{\gamma}\|_{W'}^{2} = 2(\|C_{\gamma_{v_{1}}}\|_{W'}^2-\langle C_{\gamma_{v_{1}}}, C_{\gamma_{v_{2}}}\rangle_{W'})+2(\|C_{\gamma_{h_{1}}}\|_{W'}^2-\langle C_{\gamma_{h_{1}}}, C_{\gamma_{h_{2}}}\rangle_{W'}) 
\end{equation*}
Since $\|C_{\gamma_{h_{1}}}\|_{W'}^2 = \|C_{\gamma_{h_{2}}}\|_{W'}^2 = \iint_{[0,1]\times [0,1]} k(\gamma(s)-\gamma(t)) \epsilon^2 ds dt$, it is straightforward that $\|C_{\gamma_{h_{1}}}\|_{W'}^2-\langle C_{\gamma_{h_{1}}}, C_{\gamma_{h_{2}}}\rangle_{W'}=O(\epsilon^2)$. Moreover, 
\begin{equation*}
 \|C_{\gamma_{v_{1}}}\|_{W'}^2-\langle C_{\gamma_{v_{1}}}, C_{\gamma_{v_{2}}}\rangle_{W'} = \iint_{[0,1]\times [0,1]} \left [ k(\gamma(s)-\gamma(t))- k(\gamma(s)-\gamma(t)-(\epsilon,0)) \right ] ds dt
\end{equation*} 
Since the kernel $k$ is assumed to be Lipschitz, 
$$\left [ k(\gamma(s)-\gamma(t))- k(\gamma(s)-\gamma(t)-(\epsilon,0)) \right ] \leq \Cte\epsilon$$ 
Therefore, we have eventually proved that $\|C_{\gamma}\|_{W'} = O(\sqrt{\epsilon})$ whereas for all $\epsilon$, $\mathcal{H}^{1}(\gamma)\geq 2$. In practical terms, this means that some meaningful structures of shapes, modeled as elements of a RKHS of currents, can vanish in this representation. We will detail, in a coming section, what important limitations this can induce in matching certain examples of curves or surfaces having such kind of structures.
\begin{figure}
\includegraphics[width=10cm]{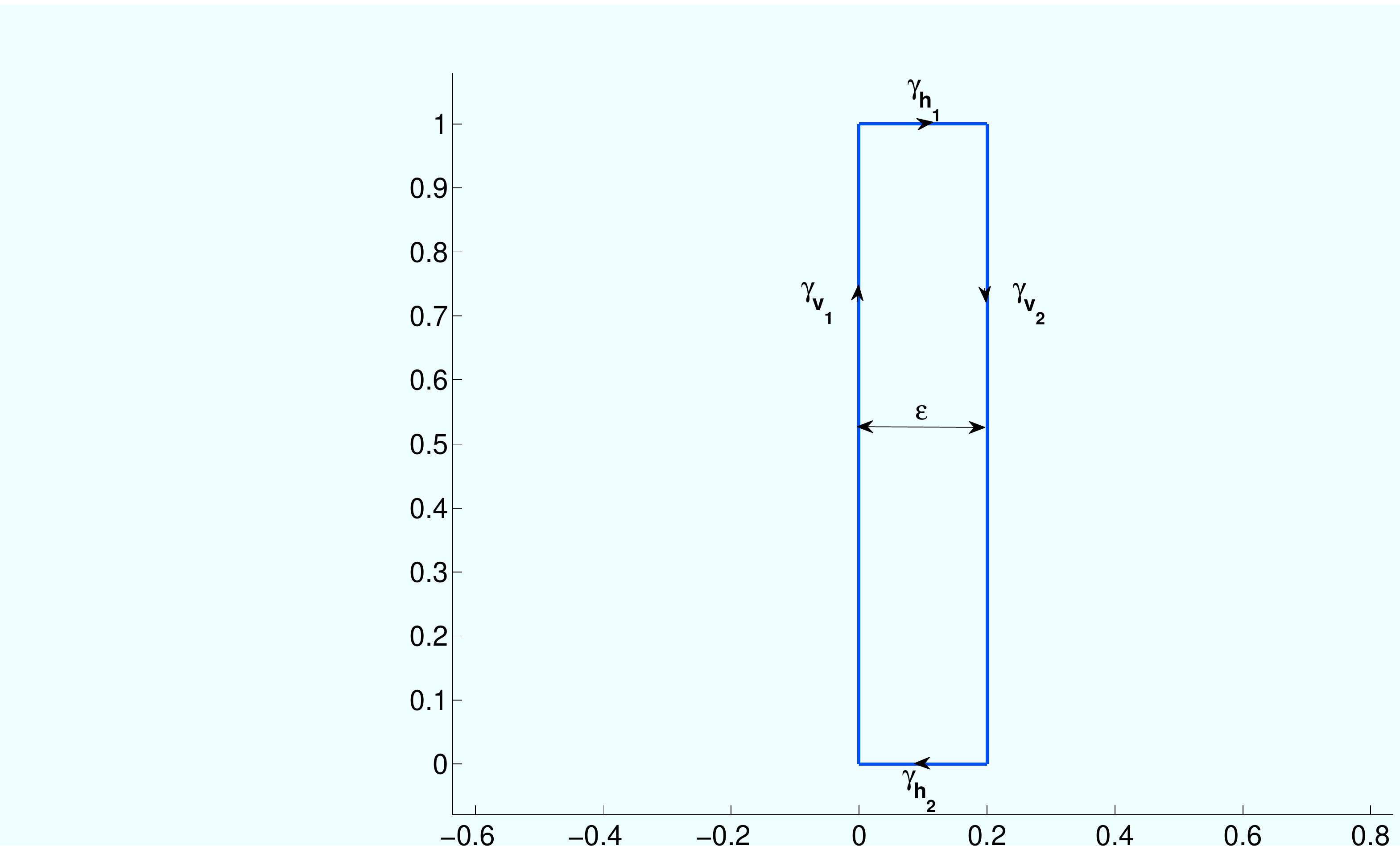}
\caption{Counter-example to the reciprocal inequality of proposition \ref{norm_current_Hausdorff}. For all $\epsilon$, $\mathcal{H}^1(\gamma)= 2(1+\epsilon) \geq 2$ whereas $\|C_{\gamma}\|_{W'}\underset{\epsilon \to 0}{\longrightarrow} 0$}\label{CEX_currents}
\end{figure}    

\section{Varifolds}
\label{sec:var}
In this section, we present varifolds first as a theoretical object, following the original work of Almgren in \cite{Almgren} and later of Allard in \cite{Allard}. We then explain in detail, in a spirit very similar to currents, the fundamental relationship between varifolds and unoriented rectifiable sets and show how this representation can be efficiently transcribed in a practical way. The construction of kernel metrics on varifolds is postponed to the next section.

\subsection{Grassmannian, varifolds~: definitions and basic properties}
\label{sec:Grass}
Varifolds have been first introduced in the context of geometric measure theory as a convenient way to address Plateau's problem of finding least area surfaces with a prescribed boundary. These developments clearly do not enter in the scope of this paper. Our purpose here is rather to connect conveniently varifolds to the framework and language of computational anatomy. Therefore, in the following, we will focus mainly on definitions of such objects and explain in what respect these definitions are computationally relevant.The basic idea behind varifolds is to represent any rectifiable set (orientable or not) as a distribution of unoriented tangent spaces spread in the embedding space $E$. As we shall see, varifolds encompass not only manifolds and rectifiable sets but more generally sets of directions in the space. Let's take again the conventions of section \ref{sec:cur}, $E$ being a vector space of dimension $n$ and $d$ an integer with $0\leq d \leq n$. What we first need is a way to represent tangent spaces of dimension $d$ in the space $E$. For this, we introduce briefly, in the following, Grassmann manifolds as it will prove useful later on. 
\begin{de}
 The Grassmann manifold of dimension $d$ in $E$, denoted $G_{d}(E)$, is the set of all $d$-dimensional subspaces of $E$. It can be identified to the quotient space of all families of $d$ independent vectors of $E$ by the equivalence relation obtained by identifying families that generate the same subspace.
\end{de}    
As a quotient space, one can show that $G_{d}(E)$ inherits a structure of compact Riemannian manifold of dimension $d(n-d)$ (cf \cite{Boothby} and \cite{Wong}). Another way to think of an element in the Grassmann manifold is to identify a subspace $V$ to the orthogonal projection on $V$, which embeds $G_{d}(E)$ into the space of linear homomorphisms $\mathcal{L}(E)$. For the following, we shall also need a practical way to think of the tangent spaces of $G_{d}(E)$ and compute variations within this manifold. If $V$ is an element of $G_{d}(E)$, we can consider $\mathcal{U}_{V}=\{W \in G_{d}(E) \ | \ W \cap V^{\bot}=\{0\} \}$, which is an open subset of the Grassmann manifold. Every element in $W \in \mathcal{U}_{V}$ can be written as $W=\{v + l(v), \ v \in V\}$ for a certain linear function $l \in \mathcal{L}(V,V^{\bot})$ and $l$ is uniquely determined by $W$, as proved in \cite{Piccione}. Thus, we get a bijective map $\psi_{V} \ : \ \mathcal{U}_{V} \rightarrow \mathcal{L}(V,V^{\bot})$. Referring again to \cite{Piccione}, $(\mathcal{U}_{V},\psi_{V})$ defines an atlas on the manifold $G_{d}(E)$ and, as a consequence, there is a natural isomorphism between the tangent space $T_{V}G_{d}(E)$ and $\mathcal{L}(V,V^{\bot})$. Now, if for $t\in ]-\epsilon,\epsilon[$, $V(t)$ is a differentiable curve on $G_{d}(E)$ with $V(0)=V_{0}$, how do we represent the derivative $V'(0)$ ? Considering any curve $v(t)$ with $v(t) \in V(t)$ for all $t$, it is easy to see that $V'(0).v(0)=p_{V_{0}^{\bot}}(v'(0))$ and this gives the definition of $V'(0)$ as an element of $\mathcal{L}(V_{0},V_{0}^{\bot})$.

Finally, let's mention a last way to embed the Grassmannian which is closer to the approach of section \ref{sec:cur}. This is formulated by the Plücker embedding property below ~:
\begin{prop}
\label{Plucker}
The following application 
\begin{eqnarray*}
 i_{P} : \ G_{d}(E) &\longrightarrow& P\left( \bigwedge^{d}E \right ) \\
 vect(v_{1},...,v_{d}) &\longmapsto& \left[ v_{1}\wedge...\wedge v_{d} \right] 
\end{eqnarray*}
where $P\left( \bigwedge^{d}E \right )$ is the real projective space of $\bigwedge^{d}E$, is an embedding. It is even a homeomorphism in the cases $d=1$ or $d=n-1$.  
\end{prop}  
\noindent As a result, we can think of an unoriented $d$-dimensional tangent space as an element of $P\left( \bigwedge^{d}E \right )$, which is a space of dimension $\binom{n}{d}-1$, much bigger in general than $G_{d}(E)$.\\
Following closely \cite{Allard}, varifolds are now defined precisely as distributions of unoriented tangent spaces in $E$~:
\begin{de}
 A $d$-dimensional varifold on $E$ is a Borel finite measure (or distribution) on the product space $E \times G_{d}(E)$, i.e an element of $C_{0}(E \times G_{d}(E))'$.
\end{de} 
Note that this differs from the original definition of varifolds given by Almgren but his definition is equivalent to the previous one as explained in the preface of \cite{Almgren}. Diracs in the space of varifolds are of the form $\delta_{(x,V)}$ with $x \in E$ and $V\in G_{d}(E)$ and act on functions of $C_{0}(E \times G_{d}(E))$ by the relation~:
\begin{equation*}
 \forall \omega \in C_{0}(E \times G_{d}(E)), \ \delta_{(x,V)}(\omega)= \omega(x,V)
\end{equation*}
In this context, a dirac consists in the data of a position $x$ in the space attached to a $d$-dimensional (non-oriented) subspace $V$ that will play the role of tangent space in the case of rectifiable sets.  

\subsection{Non-oriented shapes as varifolds}
In section \ref{sec:cur}, we have seen that oriented rectifiable sets of dimension $d$ are canonically represented as $d$-currents. In the case of non-oriented shapes, the right notion is precisely the one of varifolds. Indeed, let $X$ be a non-oriented rectifiable set of $E$ of dimension $d$. To $X$, one can associate a varifold $\mu_{X}$, which, in the measure point of view, is given by~: $\mu_{X}(A) =  \mathcal{H}^{d}\left( \{ x \in X \ | \ (x,T_{x}X) \in A \} \right)$ for all Borel subset $A \subset E \times G_{d}(E)$. By Riesz representation theorem, $\mu_{X}$ has its unique equivalent in terms of continuous linear form on $C_{0}(E \times G_{d}(E))$. For all $\omega \in C_{0}(E \times G_{d}(E))$~:
\begin{equation}
\label{eq:shape_varifold1}
 \mu_{X}(\omega)= \int_{E \times G_{d}(E)} \omega(x,V) d\mu_{X}(x,V)  
\end{equation}
Such special varifolds are called rectifiable varifolds. For more explicit expression, assume that $X$ is a smooth compact submanifold of $E$ with a parametrization $\gamma : U\rightarrow E$ with $U$ an open subset of $\mathbb{R}^{d}$, then for all $\omega \in C_{0}(E \times G_{d}(E))$ 
\begin{equation}
\label{eq:shape_varifold2}
 \mu_{X}(\omega)= \int_{U} \omega(\gamma(u),\eta(u)) |\gamma'(u)| du  
\end{equation}
with the notations $\gamma'(u)=\bigwedge_{i=1}^{d} \partial \gamma / \partial u_{i} \in \bigwedge^{d}(E)$, \ $\eta(u)=[\gamma'(u)] \in G_{d}(E)$. Note that the integral of equation (\ref{eq:shape_varifold2}) is, as expected, independent of \textbf{any reparametrization} of $X$, positively or negatively oriented.  
Now, in discrete geometry, shapes are given as polygonal sets of points which are also encompassed in the category of rectifiable subsets and therefore representable as varifolds. In the same spirit as section \ref{sec:cur}, any mesh set can be coded as a finite sum of dirac varifolds of the form $\sum_{i=1}^{n} c_{i}.\delta_{(p_{i},V_{i})}$. Again, $n$ is the number of cells of the mesh, $p_{i} \in E$ are the centers of each cell, $V_{i} \in G_{d}(E)$ the unoriented tangent space to the shape at point $p_{i}$ and $c_{i}\in \mathbb{R}_{+}^{*}$ the volume of the $i$-th cell. More specifically, let's examine the two most usual examples in practice. First, consider a curve $X$ given as a set of points $\{x_{k}\}_{k=1..N}$ and a set of edges $(f_{i}^{1},f_{i}^{2}) \in \{1,..,N\}^2$ for $i=1,..,n$. Then, with the previous conventions~:
\begin{equation}
\label{eq:discrete_curve_varifold}
\left\{
  \begin{array}[h]{l}
 p_{i}=\dfrac{x_{f_{i}^{1}}+x_{f_{i}^{2}}}{2} \\
\\
 V_{i}=(x_{f_{i}^{1}} x_{f_{i}^{2}}) \\
\\
 c_{i}=\| \overrightarrow{x_{f_{i}^{1}} x_{f_{i}^{2}}} \|\,.
  \end{array}\right.
\end{equation}
Note that all the previous equations remain unchanged if for any $i$ we take $(f_{i}^{2},f_{i}^{1})$ instead of $(f_{i}^{1},f_{i}^{2})$ as a face, which is consistent with the idea of non-orientation. In a similar way, if $X$ is now for instance a triangulated surface in $E=\mathbb{R}^3$ given by a set of points $\{x_{k}\}_{k=1..N}$ with a set of $n$ triangles $(f_{i}^{1},f_{i}^{2},f_{i}^{3}) \in \{1,..,N\}^3$ for $i=1,..,n$ then we get~:
\begin{equation}
\label{eq:discrete_surface_varifold}
\left\{
  \begin{array}[h]{l}
 p_{i}=\dfrac{x_{f_{i}^{1}}+x_{f_{i}^{2}}+x_{f_{i}^{3}}}{3} \\
\\
 V_{i}= \left [ \overrightarrow{x_{f_{i}^{1}} x_{f_{i}^{2}}},\overrightarrow{x_{f_{i}^{1}} x_{f_{i}^{3}}} \right ] \\
\\
 c_{i}=\dfrac{1}{2}.\| \overrightarrow{x_{f_{i}^{1}} x_{f_{i}^{2}}} \wedge \overrightarrow{x_{f_{i}^{1}} x_{f_{i}^{3}}} \|\,.
  \end{array}\right.
\end{equation}
Here $\left [ \overrightarrow{x_{f_{i}^{1}} x_{f_{i}^{2}}},\overrightarrow{x_{f_{i}^{1}} x_{f_{i}^{3}}} \right ]$ denotes the $2-$dimensional space generated by $\overrightarrow{x_{f_{i}^{1}} x_{f_{i}^{2}}}$ and $\overrightarrow{x_{f_{i}^{1}} x_{f_{i}^{3}}}$. Again, all these equations are not dependent on the orientation given to each triangle. \\
A last important point for the following is to be able to express the transport of varifolds by diffeomorphism in a way that is compatible with the transport of shapes. Let's fix a varifold $\mu \in C_{0}(E\times G_{d}(E))'$ and $\phi \in \Diff(E)$, we define the transport $\phi_{\ast} \mu$ of $\mu$ by $\phi$ by pull-back and push-forward operations~:
\begin{equation}
\label{eq:push_var}
 \forall \omega \in C_{0}(E\times G_{d}(E)), \ \left ( \phi_{\ast} \mu \right )(\omega)= \mu \left ( \phi^{\ast}\omega \right )
\end{equation}
where $\phi^{\ast}\omega$ is the pull-back of $\omega$ by $\phi$. For any $x \in E$ and $V \in G_{d}(E)$ a $d$-dimensional subspace with orthonormal basis $(u_{1},...,u_{d})$, $\phi^{\ast}\omega$ is defined by the relation~:
\begin{equation}
\label{eq:pull_var}
 \left ( \phi^{\ast}\omega \right )(x,V)=\left | d_{x}\phi(u_{1})\wedge...\wedge d_{x}\phi(u_{d}) \right |\,\omega(\phi(x),d_{x}\phi.V)\,.
\end{equation}
In the last equation, $d_{x}\phi.V$ denotes the element of $G_{d}(E)$ that is the image of $V$ by $d_{x}\phi$, the term $|d_{x}\phi(u_{1})\wedge...\wedge d_{x}\phi(u_{d})|$ is the $d$-dimensional Jacobian of $\phi$ on the subspace $V$ and represents the local change of $d$-dimensional volume of the transformation. With equations (\ref{eq:push_var}) and (\ref{eq:pull_var}), it is straightforward to compute the transport of a dirac $\delta_{(x,V)}$~:
\begin{equation}
\label{eq:transport_dirac}
 \phi_{\ast} \delta_{(x,V)} = \left | d_{x}\phi(u_{1})\wedge...\wedge d_{x}\phi(u_{d}) \right |.\delta_{(\phi(x),d_{x}\phi.V)}
\end{equation}
if $(u_{1},...,u_{d})$ is an orthonormal basis of $V$. Now, with the previous definitions, we can show that the varifold representation of a rectifiable set commutes with the transport by diffeomorphism in the sense given by the following proposition~:
\begin{prop}
\label{prop:commutation}
If $X$ is a $d$-dimensional rectifiable subset of $E$ and $\phi \in \Diff(E)$, then~:
\begin{equation*}
 \phi_{\ast} \mu_{X} = \mu_{\phi(X)}\,.
\end{equation*}
\end{prop}
\begin{proof}
If $X$ is a rectifiable subset of $E$, there exists (cf \cite{Allard}) at almost every point of $X$ a tangent space $T_{x}X$ and we have for all $\omega$ ~:
\begin{equation*}
 \mu_{X}(\omega)= \int_{X} \omega(x,T_{x}X) d \mathcal{H}^{d}(x)
\end{equation*}
and thus 
\begin{eqnarray*}
 \phi_{\ast}\mu_{X}(\omega) &=& \mu_{X}(\phi^{\ast}\omega) \\
 &=& \int_{X} (\phi^{\ast}\omega)(x,T_{x}X) d \mathcal{H}^{d}(x) \\
 &=& \int_{X} \omega(\phi(x),d_{x}\phi(T_{x}X)) |d_{x}\phi.T_{x}X| d \mathcal{H}^{d}(x)
\end{eqnarray*}
where $|d_{x}\phi.T_{x}X|$ is the $d$-dimensional Jacobian of $\phi$ along the tangent space $T_{x}X$, as previously. For almost all $x \in X$, $d_{x}\phi(T_{x}X)$ is a tangent space of $\phi(X)$ and by the generalization of the change of variables for rectifiable subsets (corollary 3.2.20 in \cite{Federer}), we obtain~: 
\begin{equation*}
 \phi_{\ast}\mu_{X}(\omega)= \int_{\phi(X)} \omega(y,T_{y}\phi(X)) d \mathcal{H}^{d}(y)\,.
\end{equation*}
Thus, for any $\omega$, $\phi_{\ast}\mu_{X}(\omega) = \mu_{\phi(X)}(\omega)$ and we have proved proposition \ref{prop:commutation}.  
\end{proof}
In conclusion to this section, we have seen that varifolds offer a satisfying framework to represent non-oriented shapes both from the continuous setting and in the computational cases of meshed curves, surfaces... Up to this point, we have shown how to represent computationally a meshed shape as a finite set of diracs, each of them carrying a local information of position, unoriented tangent space and local volume and derive the equations of varifolds' transport by deformation. In the purpose of adapting large deformation matching to varifolds (section \ref{sec:var_LDDMM}), an appropriate metric still needs to be defined, a topic that is discussed in the next section.

\section{Kernel metrics on varifolds}
\label{sec:var_kernel}
The use of reproducing kernels is a fundamental step when working with currents because it provides regularized metrics for the comparison of shapes which has the additional advantage of having an underlying Hilbert space structure. We have also mentioned in section \ref{sec:cur} that kernels are particularly well-fitted to compute distances between discretized shapes because of the simple expression of dot product between two diracs. All these reasons motivate a similar approach in the treatment of unoriented shapes through the varifold setting that has been presented. It's worth mentioning that other approaches could be possible, by working with more general Riemannian metrics on distribution of tangent spaces, generalizing for instance what is done for the Grassmann manifold in \cite{Absil}. We argue however that RKHS metrics are very convenient from a computational point of view, especially in our applications to shape matching and analysis. In this section, we propose a generic and simple way to build relevant reproducing kernels on the space of varifolds, by basically making tensor products of kernels on $E$ and on $G_{d}(E)$. 

\subsection{Kernels on the Grassmann manifold}
\label{sec:kernel_Grassmann}
We start by the construction of kernels on $G_{d}(E)$. This problem is not totally new since examples of such kernels have been provided in several recent works related to machine learning, notably in \cite{Hamm} and \cite{Wolf}. Such constructions are usually based on the notion of \textit{principle angles} between two subspaces. An alternative way, which has the advantage of offering a very wide class of possible kernels on the Grassmannian, is to simply use the embedding of $G_{d}(E)$ into $\mathcal{L}(E)$ mentioned in section \ref{sec:var}. It identifies any subspace $V\in G_{d}(E)$ with the orthogonal projection on $V$, which we denote by $p_{V}$. As a consequence, one can induce straightforwardly a kernel on $G_{d}(E)$ by restriction of a positive kernel defined on $\mathcal{L}(E)$. Since $\mathcal{L}(E)$ is a finite-dimensional vector space, there are no difficulties in defining such kernels. Typical choices are given by~:
\begin{equation}
\label{eq:kernel_CB1}
 k_{CB}(V,W)=\langle p_{V},p_{W} \rangle^{k}
\end{equation}
with $k \in \mathbb{N}^{\ast}$ and $\langle .,.\rangle$ the usual Frobenius metric in $\mathcal{L}(E)$. More generally, any function $P(\langle p_{V},p_{W} \rangle)$, $P$ being a polynomial with positive coefficients, is a positive kernel on $\mathcal{L}(E)$. Other important possibilities are the kernels induced by Gaussians in $\mathcal{L}(E)$, namely~:
\begin{equation}
\label{eq:kernel_G1}
 k_{G}(V,W)=e^{-\dfrac{|p_{V}-p_{W}|^{2}}{\sigma^{2}}}
\end{equation}
Such kernels allow the comparison of subspaces with respect to a certain scale given by parameter $\sigma$. Many other could be defined with this method. It's also possible to express these kernels using the principle angles. If $V$ and $W$ are two $d$-dimensional subspaces of $E$, the $d$ principle angles $\theta_{1},..,\theta_{d}$ between them are defined recursively by the relations~:
\begin{eqnarray*}
 \cos(\theta_{k}) &=& \max_{v_{k}\in V} \max_{w_{k} \in W} \langle \frac{v_{k}}{|v_{k}|},\frac{w_{k}}{|w_{k}|} \rangle \ with \\
 & & \forall i \in \{1,..,k-1\}, \ v_{k}\bot v_{i}, \ w_{k}\bot w_{i}
\end{eqnarray*}
Thus, there exists orthonormal frames $(v_{1},..,v_{d})$ and $(w_{1},..,w_{d})$ of $V$ and $W$ such that $\cos(\theta_{i})=\langle v_{i}, w_{i} \rangle$. As easily seen,  see e.g. \cite{Wong}, there are at most $\min(d,n-d)$ non-zero principle angles and they completely determine the relative position between the two subspaces. Now it's an easy verification that~: 
\begin{equation*}
 \langle p_{V},p_{W} \rangle = 2\sum_{i=1}^{d} \langle v_{i},w_{i} \rangle^{2} = 2\sum_{i=1}^{d} \cos^{2}(\theta_{i})
\end{equation*}
It results that, up to a scaling factor, the kernel of equation (\ref{eq:kernel_CB1}) has the following expression~:
\begin{equation}
\label{eq:kernel_CB2}
 k_{CB}(V,W)=\left( \sum_{i=1}^{d} \cos^{2}(\theta_{i}) \right )^{k}
\end{equation}
For $k=1$, we see that the kernel we obtain is precisely the classical Cauchy-Binet kernel on the Grassmannian, as in \cite{Hamm}. The Gaussian kernel also has the alternative expression~:
\begin{equation}
\label{eq:kernel_G2}
 k_{G}(V,W)=e^{-\frac{4}{\sigma^{2}}.\sum_{i=1}^{d} (1-\cos^{2}(\theta_{i})) }=\prod_{i=1}^{d} e^{-\frac{4}{\sigma^{2}}\sin^{2}(\theta_{i})}
\end{equation}
This formulation with principle angles offers, in our sense, a nicely interpretable way to understand the behavior of such kernels with respect to the subspaces' relative position. The general construction we have proposed thus provides a wide variety of induced kernels on $G_{d}(E)$ that are effectively computable (either by the expression of the projection matrix or by computing the principle angles).
\begin{remark}
 We also want to mention, without entering into details, another possible way of building kernels on $G_{d}(E)$, relying on the Plücker embedding of section \ref{sec:Grass} instead of the identification with projectors. Indeed, one can induce as previously a kernel on the Grassmann manifold from a kernel on $P(\bigwedge^{d}E)$. Since this projective space can be also identified to the unit sphere of $\bigwedge^{d}(E)$ quotiented by the reorientation group $\{\pm 1\}$, it is not difficult to provide kernels on $P(\bigwedge^{d}E)$ by defining kernels on the vector space $\bigwedge^{d}E$, invariant with respect to the action of $\{\pm 1\}$. For curves or hypersurfaces, the Plücker embedding being a homeomorphism, this approach is relevant and provides kernels quite similar to our previous method. However, for high dimension and codimension, we see that the embedding space becomes very high-dimensional $\binom{n}{d}-1$ compared to the actual dimension of $G_{d}(E)$ and we believe that this can be a limiting factor both from a theoretical and computational point of view.   
\end{remark}

\subsection{Construction of kernels on varifolds}
\label{sec:var_kernel_tensprod} 
The issue of kernels on $G_{d}(E)$ being addressed, we now move to the case of varifolds. Since we are considering functions defined on a product space, a natural way to build appropriate kernels is the tensor product trick. We first remind a well-known general property from reproducing kernels' theory:
\begin{lemme}
\label{lemma:tensor_product}
 Let $A$ and $B$ be two sets and $k_{A}$, $k_{B}$ two positive real kernels respectively on $A$ and $B$. Then, the tensor product $k_{A}\otimes k_{B}$ defined by~:
 \begin{equation*}
  k_{A} \otimes k_{B} \left ( (a,b),(a',b') \right )= k_{A}(a,a')\,k_{B}(b,b')
 \end{equation*}
  is a positive real kernel on $A \times B$.
\end{lemme}
\begin{proof}
 For any $N \in \mathbb{N}^{\ast}$, $a_{1},..,a_{N} \in A$ and $b_{1},...,b_{N} \in B$ then, since $k_{A}$ and $k_{B}$ are positive kernels, the matrices $\left [ k_{A}(a_{i},a_{j}) \right ]_{i,j}$ and $\left [ k_{B}(b_{i},b_{j}) \right ]_{i,j}$ are positive by definition. The matrix $\left [ k_{A}\otimes k_{B}((a_i,b_i),(a_j,b_j)) \right ]_{i,j}$ being the Hadamard product of the two previous matrices, it follows from Schur product theorem that it is also positive. Thus, $k_{A}\otimes k_{B}$ is a positive kernel.   
\end{proof}
We remind that to the kernel $k_{A} \otimes k_{B}$ is associated a unique Hilbert space of real-valued functions on $A \times B$ (the RKHS of the kernel) on which all linear forms $\delta_{(a,b)} : f \ \mapsto f(a,b)$ are continuous. Now, going back to the case of varifolds, the following holds~:
\begin{prop}
 Assume that we are given a positive real kernel $k_{e}$ on the space $E$ such that $k_{e}$ is continuous, bounded and for all $x \in E$, the function $k_{e}(x,.)$ vanishes at infinity. Assume that a second kernel $k_{t}$ is defined on the manifold $G_{d}(E)$ and is also continuous. Then the RKHS $W$ associated to the kernel $k_{e} \otimes k_{t}$ is continuously embedded into the space $C_{0}(E\times G_{d}(E))$. 
 \label{prop:kernel_var}
\end{prop}
\begin{proof}
By definition~:
  \begin{equation}
\label{eq:def_newkernel}
  k \left ((x,V),(y,\tilde{V}) \right ) = k_{e}(x,y)\,k_{t}(V,\tilde{V})
 \end{equation} 
 and so, thanks to the assumptions on the kernels, $k((x,V),.)$ is continuous on $E \times G_{d}(E)$ and belongs to $C_{0}(E\times G_{d}(E))$. $W$ being the Hilbert space generated by the functions $k((x,V),.)$, we have $W \subset C_{0}(E\times G_{d}(E))$. Moreover, if $\omega \in W$~:
 \begin{equation*}
  \omega(x,V)=\delta_{(x,V)}(\omega)=\langle k((x,V),.), \omega \rangle_{W}
 \end{equation*}
With Cauchy-Schwarz inequality~: $|\omega(x,V)| \leq \|k((x,V),.) \|_{W}.\|\omega\|_{W}$. In addition, $\|k((x,V),.) \|_{W}=\sqrt{k((x,V),(x,V))}$ and both kernels $k_{e}$ and $k_{t}$ are bounded so that $k$ is also bounded. We conclude that $|\omega|_{\infty} \leq \sqrt{|k|_{\infty}}.\|\omega\|_{W}$, which proves that the inclusion embedding $\imath: \ W\hookrightarrow C_{0}(E\times G_{d}(E))$ is indeed continuous.
\end{proof}
Consequently, there exists a continuous mapping of the space of varifolds $C_{0}(E\times G_{d}(E))'$ into the dual of $W$. Just as for currents, we can compare varifolds and unoriented subsets through the Hilbert norm of $W'$. As a result of the kernel property, if $x_{1},x_{2} \in E$ and $V_{1},V_{2} \in G_{d}(E)$ ~:
 \begin{equation}
\label{eq:dot_product_diracs}
  \langle \delta_{(x_{1},V_{1})}, \delta_{(x_{2},V_{2})} \rangle_{W'}=k_{e}(x_{1},x_{2})\,k_{t}(V_{1},V_{2})
 \end{equation}
We have therefore a generic way to build kernels for varifolds which are \textit{separable} since such kernels are tensor products of a kernel on the ambient space $E$ with a kernel on the set of all tangent spaces $G_{d}(E)$. Building kernels on the euclidean space $E$ raises no difficulties (examples were already given in section \ref{sec:cur}) and we have presented a way to build kernels on $G_{d}(E)$. However, it still remains unclear whether the dual application $i^*: \ C_{0}(E\times G_{d}(E))'\rightarrow W'$ is always an embedding. In general, this is actually not the case~: although $W\hookrightarrow C_{0}(E\times G_{d}(E))$ is an embedding, the dual application need not be injective. For instance, choosing the Cauchy-Binet kernel of equation (\ref{eq:kernel_CB1}) on the Grassmannian makes $i^*$ not injective. The fact that $i^*$ is an embedding is called the \textbf{$C_{0}$-universality} property of the kernel and, as proved in \cite{Carmeli}, it is equivalent to the property of $W$ being dense in $C_{0}(E\times G_{d}(E))$. With respect to varifolds' kernels, we can still show a general but weaker result~: under a few assumptions, $i^*$ is injective on the set of finite unions of submanifolds. The result is the following~:
\begin{prop}
 Let $k=k_{e}\otimes k_{t}$ be a kernel as in proposition \ref{prop:kernel_var}. Assume that kernel $k_{e}$ is $C_{0}$-universal and that the kernel $k_{t}$ is such that $k_{t}(V,V)> 0$ for all $V \in G_{d}(E)$. Let $X=\bigcup_{i=1}^{N} X_{i}$ and $Y=\bigcup_{j=1}^{M} Y_{j}$ be two finite union of compact $d$-dimensional submanifolds of $E$. If $\| \mu_{X} - \mu_{Y} \|_{W'}=0$ then $X=Y$.
 \label{prop:injectivity_submanifold}
\end{prop}
\begin{proof}
 We will denote by $W_{e}$ and $W_{t}$ the RKHS associated to kernels $k_{e}$ and $k_{t}$. Let's start by the case where $X$ is a single submanifold of $E$. If $\| \mu_{X} - \mu_{Y} \|_{W'}=0$ and $X\neq Y$, one can find $x_{0} \in \mathring{X} \backslash Y$ and there exists $r>0$ such that $B(x_{0},r) \cap Y = \emptyset$. Let $a \in C_{0}(E,\mathbb{R})$ be any continuous function such that the support of $a$ is included in $B(x_{0},r)$. Let $\omega_{t}=k_{t}(T_{x_{0}}X,.) \in W_{t}$. For all functions $\omega_{e} \in W_{e}$, $\omega_{e} \otimes \omega_{t} \in W$ and thus $(\mu_{X}-\mu_{Y})(\omega_{e} \otimes \omega_{t})=0$, i.e
 \begin{equation*}
  \int_{X} \omega_{e}(x)\omega_{t}(T_{x}X) d\mathcal{H}^{d}(x) - \int_{Y} \omega_{e}(y)\omega_{t}(T_{y}Y) d\mathcal{H}^{d}(y) = 0
 \end{equation*}
Now, since $W_{e}$ is dense in $C_{0}^{0}(E,\mathbb{R})$, we can approximate uniformly $a$ by functions in $W_{e}$. By uniform convergence in the previous integral and the fact that $a=0$ in $Y$, we get~:
 \begin{equation*}
  \int_{X} a(x)\omega_{t}(T_{x}X) d\mathcal{H}^{d}(x) = 0
 \end{equation*}
This holds for all continuous functions $a$ supported in $B(x_{0},r)$, which is clearly impossible since $x\mapsto \omega_{t}(T_{x}X)$ is continuous and $\omega_{t}(T_{x_{0}}X)=k_{t}(T_{x_{0}}X,T_{x_{0}}X) \neq 0$. \\
If $X$ is now a finite reunion of compact submanifolds and $X\neq Y$ then, as previously, we can find $x_{0} \in \bigcup_{i=1}^{N} \mathring{X_{i}}$ with $B(x_{0},r_{0}) \cap Y = \emptyset$. $B(x_{0},r_{0}) \cap X$ is itself a non-empty finite reunion of submanifolds. In order to get a similar proof as for the one submanifold case, we need to show the following~: there exists a point $\hat{x} \in B(x_{0},r_{0}) \cap X$, $\rho >0$ with $B(\hat{x},\rho) \subset B(x_{0},r_{0})$ and $j \in \{1,..,N\}$ such that for all function $\omega$ supported in $B(\hat{x},\rho) \times G_{d}(E)$, $\mu_{X}(\omega)=\mu_{X_{j}}(\omega)$. This can be proved recursively. \\
Assuming the result for $N-1$, we can find $x_{1}$, $r_{1}$,$j_{1}$ with $B(x_{1},r_{1}) \subset B(x_{0},r_{0})$ and $\mu_{\bigcup_{i=1}^{N-1}X_{i}}=\mu_{X_{j_{1}}}$ for functions supported in $B(x_{1},r_{1})$. Now, there are two distinct cases~: either $B(x_{1},r_{1}) \cap X_{j_{1}} = B(x_{1},r_{1}) \cap X_{N}$ or we can assume for instance that there exists $x_{2} \in B(x_{1},r_{1}) \cap X_{j_{1}}$, $r_{2}>0$ such that $B(x_{2},r_{2}) \subset B(x_{1},r_{1})$ and $B(x_{2},r_{2}) \cap X_{N}= \emptyset$. In the first case, since $X_{N}$ coincides with $X_{j_{1}}$ on $B(x_{1},r_{1})$, we have still $\mu_{\bigcup_{i=1}^{N}X_{i}}=\mu_{X_{j_{1}}}$ for functions supported in $B(x_{1},r_{1}) \times G_{d}(E)$ and we take $\hat{x}=x_{1}$, $\rho=r_{1}$, $j=j_{1}$. In the second case, one can take $\hat{x}=x_{2}$, $\rho=r_{2}$, $j=j_{1}$ and the result holds as well. The rest of the proof is then exactly similar to the case of a unique submanifold.
\end{proof}
\noindent Even though there is no general injectivity for the dual application, proposition \ref{prop:injectivity_submanifold} ensures that the metrics we use are at least able to distinguish finite reunion of submanifolds. We believe that the result could be generalized to rectifiable subsets of $E$ but the proof might be more involved and technical. \\
In the specific case of the Gaussian kernel on $G_{d}(E)$ given in section \ref{sec:kernel_Grassmann}, we can actually recover the injectivity of $i^*$ on $C_{0}(E\times G_{d}(E))'$ itself, which is stated in the next proposition~:
\begin{prop}
 If $k=k_{e}\otimes k_{G}$ with $k_{e}$ a $C_{0}$-universal kernel on $E$ and $k_{G}$ the restriction of a Gaussian kernel on $\mathcal{L}(E)$ given in equation (\ref{eq:kernel_G1}), then $k$ is a $C_{0}$-universal kernel on $E \times G_{d}(E)$.
 \label{prop:universal_kernel}
\end{prop}
\begin{proof}
 The proof is mainly based on two results about RKHS. The first one is the very well-known property that a Gaussian kernel on a finite dimensional vector space is $C_{0}$-universal. In our context, the Gaussian kernel $K$ defined on $\mathcal{L}(E) \times \mathcal{L}(E)$ by $K(l_{1},l_{2})= e^{-\frac{|l_{1}-l_{2}|^{2}}{\sigma^{2}}}$ is thus $C_{0}$-universal. Now, the kernel $k_{G}$ that we have defined in equation (\ref{eq:kernel_G1}) is the restriction of $K$ to the subset of all orthogonal projections on $d$-dimensional subspaces (which is identified to $G_{d}(E)$). This subset is closed in $\mathcal{L}(E)$ and it is proved in \cite{Carmeli} (corollary 3) that $k_{G}$ is then a $C_{0}$-universal kernel. If $W_{e}\subset C_{0}(E,\mathbb{R})$ and $W_{t}\subset C_{0}(G_{d}(E),\mathbb{R})$ are the RKHS corresponding to kernels $k_{e}$ and $k_{G}$, it results that we have $W_{e}$ dense in $C_{0}(E,\mathbb{R})$ (because $k_{e}$ is assumed to be $C_{0}$-universal) and $W_{t}$ dense in $C_{0}(G_{d}(E),\mathbb{R})$. It is then clear that $W$, which is the completion of $W_{e} \otimes W_{t}$, is dense in $C_{0}(E \times G_{d}(E),\mathbb{R})$.
\end{proof}
\noindent In summary, we have explained in this subsection how to define Hilbert metrics on varifolds that are computed from tensor products between kernels on $E$ and kernels on the Grassmann manifold. This does not provide necessarily real distances on varifolds because of the possible non-injectivity of the application between the space of varifolds and the dual of the RKHS $W$. Yet, we have shown that for a very wide class of such kernels, the resulting distances are separating finite unions of submanifolds and we have argued in favor of Gaussian kernels on the Grassmann manifold for which we obtain the real injectivity and thus distances on the space of all varifolds. We will elaborate a little more on the interest of such Gaussian kernels in specific situations in the next subsection (cf figure \ref{dist_CB_Gaussian}).

\subsection{Properties of RKHS norms on varifolds}
\label{sec:kernel_var_properties}
There are several interesting remarks and properties that we can state about the previously defined norms, notably if we compare it to the framework of currents that was exposed in the first section. First of all, let's consider the limit case of an infinite scale for the Gaussian kernel on the space $E$, i.e $k_{e}(x,y)=Id_{E}$. If $X$ is any $d$-dimensional submanifold of $E$, the norm of the current $C_{X}$ becomes~:
 \begin{equation*}
  \|C_{X} \|_{W'}^2 = \iint_{X \times X} \langle \xi(x), \xi(y) \rangle = \langle \int_{X} \xi(x), \int_{X} \xi(x) \rangle
 \end{equation*}
Now, from Stokes' formula, $\int_{X} \xi(x)$ is a term that only depends on the boundary of $X$. Therefore, at large scale, RKHS norms on currents are only sensitive to the boundaries of objects, and all shapes with no boundary vanishes in this representation. \\
In the varifold case, the behavior is fundamentally different. At large scale for the kernel $k_{e}$, the RKHS distance becomes a distance between the distributions of tangent space directions depending on the kernel $k_{t}$. More specifically~:
 \begin{eqnarray*}
  \|\mu_{X} \|_{W'}^2 &=& \iint_{X \times X} k_{t}(T_{x}X,T_{y}X) d\mathcal{H}^{d}(x) d\mathcal{H}^{d}(y) \\
  &=& \iint_{G_{d}(E) \times G_{d}(E)} k_{t}(u,v) d\nu_{X}(u) d\nu_{X}(v)
 \end{eqnarray*}
 where $\nu_{X}:=\mathcal{H}^{d} \circ Tan^{-1}$ is the image measure of $\mathcal{H}^{d}$ by the application $Tan: X\rightarrow G_{d}(E), \ x\mapsto T_{x}X$. In other words, at infinite scale for $k_{e}$, $\|\mu_{X} \|_{W'}^2$ represents a metric on the distribution of the non-oriented tangent spaces to the shape $X$, which we can see, in a certain way, as a histogram metric on $G_{d}(E)$. Therefore, even at infinite scale for $k_{e}$, something of the shape still remains in the varifold representation contrarily to currents. In those situations, the choice of the kernel on the Grassmann part of varifolds is quite decisive. In particular, the $C_{0}$-universality issue discussed previously becomes important. We show an illustration of such phenomenon in figure \ref{dist_CB_Gaussian} for the two kernels explicitly given in section \ref{sec:kernel_Grassmann}~: the Cauchy-Binet kernel of equation (\ref{eq:kernel_CB1}), which is not $C_{0}$-universal and the Gaussian kernel of equation (\ref{eq:kernel_G1}), which was proved to be. The conclusion that we can be drawn in general is that Gaussian kernels are obviously more effective in situations of multiple directions crossing at nearby points, as might appear when treating fibers or tree-like structures.    \\
 \\
\begin{figure}
\begin{tabular}{cc}
\includegraphics[width=5cm,height=5cm]{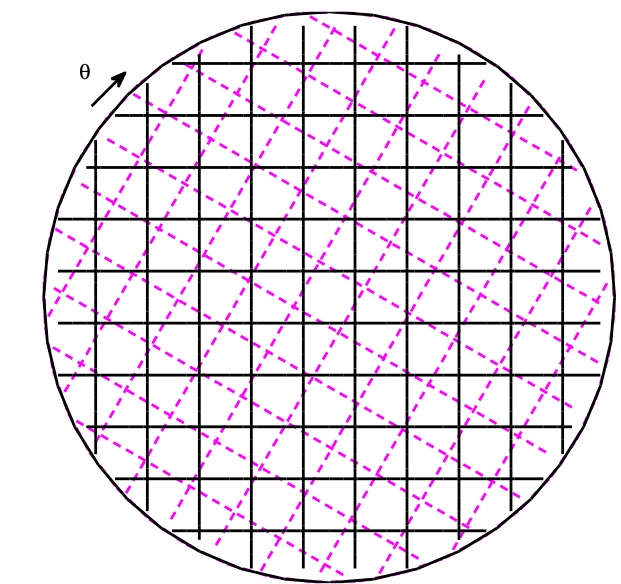} & \includegraphics[width=7cm,height=5cm]{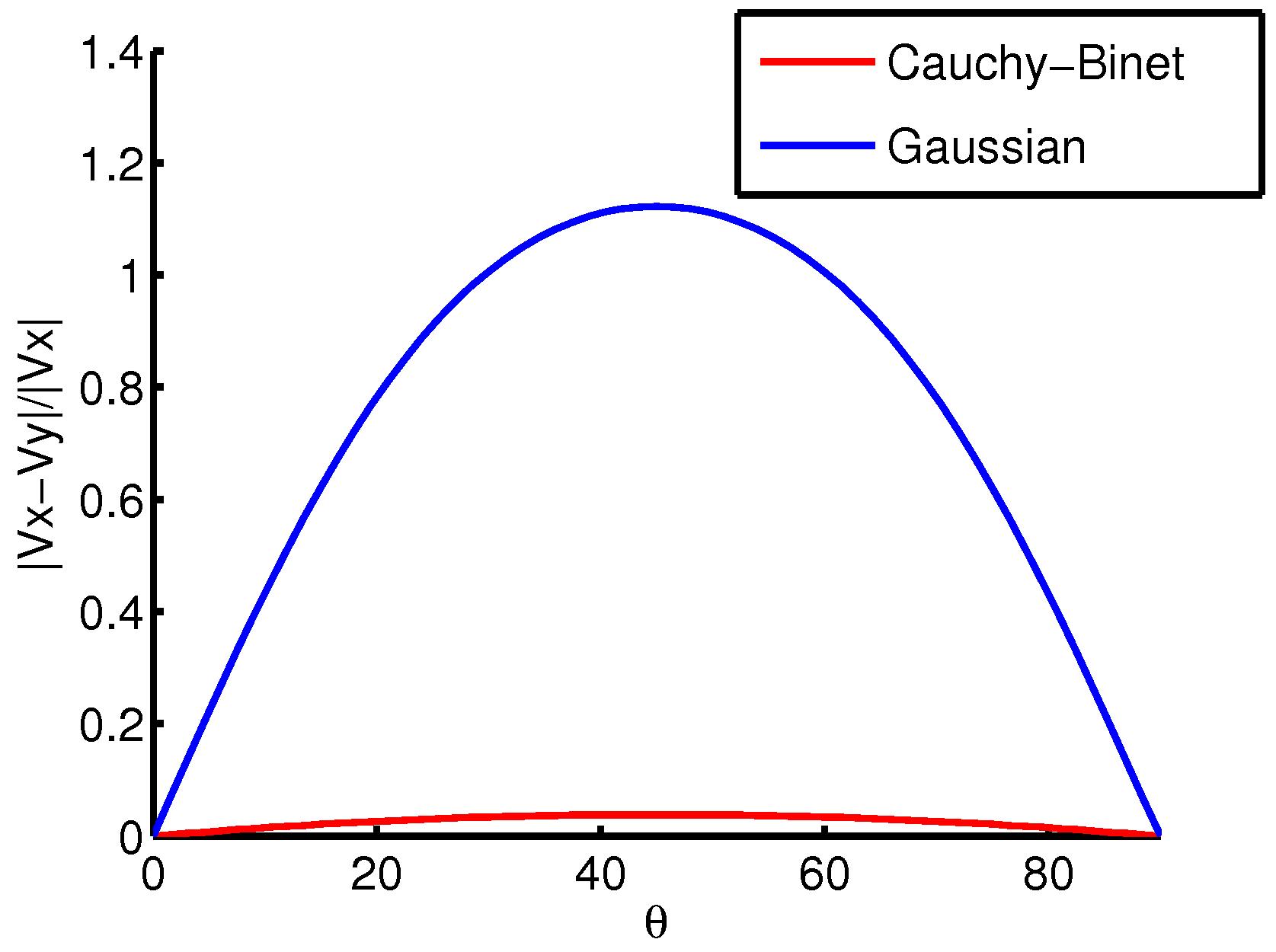}
\end{tabular}
\caption{Computation of varifolds' distance between a 2D grid-like shape (in black) and its rotated version (in magenta) for various rotation angles and for both the Cauchy-Binet kernel and the Gaussian kernel on the Grassmann manifold, all in the case of a very large-scale kernel on the space $E$. The graph displays the ratio between the varifold distances and the norm of the original shape as a function of the rotation angle. With the Cauchy-Binet kernel, the values remain very small for all angles, and it is nearly unable to distinguish both shapes whereas the Gaussian kernel shows an expected behavior with a maximum distance for a 45$^{\circ}$ angle.}
\label{dist_CB_Gaussian}
\end{figure}
 
\noindent We now come the problem raised by our discussion related to figure \ref{CEX_currents}. We show in which precise sense RKHS metrics on varifolds successfully avoid the cancellation phenomenon of currents. This is formulated in the theorem below and its corollary~:
 \begin{theo}
 \label{theo:principal}
  Let $\mu$ be a positive $d$-dimensional varifold of $E$ such that $\mu \in W'$ and $Supp(\mu)\subset B(0,1)\times G_{d}(E)$. Let $k$ be a reproducing kernel on $E\times G_{d}(E)$ and $W$ its RKHS, defined like in section \ref{sec:var_kernel_tensprod} by $k((x,u),(y,v))=k_{e}(x,y).k_{t}(u,v)$ where $k_{e}$ is a reproducing kernel on $E$ and $k_{t}$ a reproducing kernel on the Grassmannian. We make the following additional assumptions~: \\
  1. $k_{e}$ is a radial scalar kernel so that we can write $k_{e}(x,y)=h_{e}(|x-y|)$. $h_{e}$ is assumed to be a continuous positive function, with $h_{e}(0)>0$.\\
  2. $k_{t}$ is a continuous and positive function on $G_{d}(E) \times G_{d}(E)$ such that for all $u \in G_{d}(E)$, $k_{t}(u,u)>0$.\\ 
  Then there exists a constant $\Cte$ independent of $\mu$ such that~:
  \begin{equation*}
  \|\mu\|_{W'} \geq \Cte\mu(E \times G_{d}(E))\,.
  \end{equation*}
 \end{theo}
 
\begin{proof}
For any varifold $\mu \in W'$, we have~:
 \begin{equation}
 \label{eq.control1}
  \|\mu\|_{W'}^2 = \iint_{(E \times G_{d}(E))^{2}} k_{e}(x,y).k_{t}(u,v) d\mu(x,u) d\mu(y,v)\,.
 \end{equation} 
For the proof of theorem \ref{theo:principal}, we shall first examine the case of a constant kernel for $k_{e}$. \\
 \textbf{Step 1:} We first assume that $k_{e}(x,y)=1$ for all $x,y$. Let's denote by $p : E \times G_{d}(E) \rightarrow G_{d}(E)$ the application $(x,u)\mapsto u$. We introduce the image measure $\nu:=\mu \circ p^{-1}$ defined on $G_{d}(E)$. Then~:
 \begin{equation*}
  \|\mu \|_{W'}^2 = \iint_{G_{d}(E) \times G_{d}(E)} k_{t}(u,v) d\nu(u) d\nu(v)\,.
 \end{equation*} 
Note that $\nu(G_{d}(E))=\mu(p^{-1}(G_{d}(E))) = \mu(E\times G_d{E})$. Now, the compact group $SO(E)$ of direct isometries of $E$ acts transitively on $G_{d}(E)$ by the relation \\
$g.\text{Span}(e_{1},..,e_{d}):=\text{Span}(g(e_{1}),..,g(e_{d}))$, making $G_{d}(E)$ a homogeneous space. Since $SO(E)$ is a compact group, we can consider its unique bi-invariant Haar measure which we denote $\lambda$ with the convention $\lambda(SO(E))=1$. If we fix a particular element $u_{0} \in G_{d}(E)$, $\lambda$ induces in turn a measure on $G_{d}(E)$ defined for all $B$ by $\lambda_{Gr}(B)=\lambda(\{g \in SO(E) \ | \ g.u_{0} \in B \})$. One can check that the right-invariance of $\lambda$ implies that $\lambda_{Gr}$ does not depend on the choice of $u_{0} \in G_{d}(E)$. In the same way, thanks to the left-invariance of $\lambda$, $\lambda_{Gr}$ is invariant by the action of $G$ i.e $\lambda_{Gr}(g.B)=\lambda_{Gr}(B)$ for all $g$. 

A second important element is that, as a homogeneous space under the action of a compact Lie group (on which there exists a bi-invariant metric), $G_d(E)$ can be equipped by projection with a left-invariant distance with respect to the action of $SO(E)$, which we will denote $d$. Let's now consider a continuous, positive $L^1$ (for the measure $\lambda_{Gr}$) function $\phi$ on $G_{d}(E)$. Since $\lambda_{Gr}(G_{d}(E))=\lambda(SO(E))=1 < + \infty$, $\phi$ is also in $L^2$ and we assume that $\|\phi\|_{2}=1$. We can assume in addition that the support $\text{Supp}(\phi)$ of $\phi$ is included in a certain ball of radius $\delta$ centered at $u_{0}$. We also introduce the regularization function $\psi$ on $G_{d}(E)\times G_{d}(E)$ defined by ~:
 \begin{equation*}
  \psi(u,v) = \int_{SO(E)} \phi(g.u) \phi(g.v) d\lambda(g)\,.
 \end{equation*}
It's obvious that $\psi$ is also a positive continuous function and Cauchy-Schwarz inequality gives~: 
 \begin{equation*}
\forall (u,v)\in G_{d}(E)^{2}, \ \psi(u,v) \leq \int_{SO(E)} \phi(g.u)^{2} d\lambda(g)\,.
 \end{equation*}
The definition and invariance properties of $\lambda_{Gr}$ show that for all $u \in G_{d}(E)$, if $\tau_{u} : g\mapsto g.u$ then, $\lambda \circ \tau_{u}^{-1}=\lambda_{Gr}$ and it results that~:
\begin{equation*}
\int_{SO(E)} \phi(g.u)^{2} d\lambda(g) = \int_{G_{d}(E)} \phi(v)^{2} d\lambda_{Gr}(v) = \|\phi\|_{2}^{2}=1\,,
 \end{equation*}
therefore $\psi(u,v) \leq 1$. In addition, if $d(u,v) \geq 2\delta$, we have $d(g.u,g.v) \geq 2\delta$ for all $g$ since $d$ is left-invariant. Therefore $d(g.u,u_{0})+d(g.v,u_{0}) \geq d(g.u,g.v) \geq 2\delta$. It results that either $d(g.u,u_{0}) \geq \delta$ or $d(g.v,u_{0}) \geq \delta$ and anyhow $\phi(g.u) \phi(g.v)=0$. Consequently, $\psi$ vanishes for $d(u,v) \geq 2\delta$. Now, since $k_{t}$ is positive  
 \begin{equation*}
  \|\mu\|_{W'}^2 \geq \iint_{G_{d}(E) \times G_{d}(E)} k_{t}(u,v) \psi(u,v) d\nu(u) d\nu(v)\,.
 \end{equation*} 
Moreover, $k_{t}$ is also continuous and $k_{t}(u,u)>0$ for all $u$ so, by compactness of $G_{d}(E)$, there exists $\alpha>0$ and an open domain $D\subset G_{d}(E) \times G_{d}(E)$ of the form $D=\{(u,v) \ | \ d(u,v)<\epsilon\}$ such that for all $u,v \in D$, $k_{t}(u,v) \geq \alpha$. We can also assume, by choosing appropriately the function $\phi$, that $\delta \leq \frac{\epsilon}{2}$ and therefore $\psi(u,v)=0$ outside $D$. It results that~:
  \begin{equation}
  \label{eq.control2}
  \|\mu\|_{W'}^2 \geq \alpha \iint_{G_{d}(E) \times G_{d}(E)} \psi(u,v) d\nu(u) d\nu(v)\,.
 \end{equation}
 Using Fubini's theorem in the last integral, we can write 
 \begin{align}
 \label{eq.control3}
  \iint_{G_{d}(E) \times G_{d}(E)} \psi(u,v) & d\nu(u) d\nu(v)\nonumber\\
& = \int_{SO(E)} \left ( \iint_{G_{d}(E) \times G_{d}(E)} \phi(g.u) \phi(g.v) d\nu(u) d\nu(v) \right ) d\lambda(g) \nonumber \\
  & = \int_{SO(E)} \left ( \int_{G_{d}(E)} \phi(g.u) d\nu(u) \right)^{2} d\lambda(g) \nonumber \\
  & \geq \left ( \int_{SO(E)} \int_{G_{d}(E)} \phi(g.u) d\nu(u) d\lambda(g) \right )^{2}
 \end{align}
 the last estimate resulting from Cauchy-Schwartz inequality and the fact that $\lambda(SO(E))=1$. Using again Fubini's theorem,  
  \begin{equation*}
  \int_{SO(E)} \int_{G_{d}(E)} \phi(g.u) d\nu(u) d\lambda(g) = \int_{G_{d}(E)} \left ( \int_{SO(E)} \phi(g.u) d\lambda(g) \right ) d\nu(u)\,.
 \end{equation*}
 Making the change of variable $\tau_{u}: \ g\mapsto g.u$ in the inside integral and using the same argument as previously, we obtain that
 \begin{equation*}
  \int_{SO(E)} \phi(g.u) d\lambda(g) = \int_{G_{d}(E)} \phi(w) d\lambda_{Gr}(w) = \|\phi\|_{1}>0
 \end{equation*}
 Inserting the previous in equation (\ref{eq.control3}), 
 \begin{equation*}
  \iint_{G_{d}(E) \times G_{d}(E)} \psi(u,v) d\nu(u) d\nu(v) \geq \|\phi\|_{1}^{2} \nu(G_{d}(E))^{2}= \|\phi\|_{1}^{2} \mu(E\times G_d{E})^{2}\,.
 \end{equation*} 
 So that we eventually obtain~:
 \begin{equation*}
  \|\mu_{X}\|_{W'}^2 \geq \alpha \|\phi\|_{1}^{2}. \mu(E\times G_d{E})^{2}
 \end{equation*} 
 and since $\beta:=\alpha \|\phi\|_{1}^{2}$ is a constant that does not depend on $\mu$, this concludes the proof in that case. \\
 \textbf{Step 2:} We now move to the proof for a general kernel $k_{e}$. From the hypotheses, the support of $\mu$ is included in $B(0,1) \times G_{d}(E)$ and $k_{e}$ is a continuous and positive radial scalar kernel on $E$. Thus, there exists a real number $\delta <1$ such that for any $x,y \in E$ with $|x-y|\leq \delta$, we have $k_{e}(x,y) \geq h_{e}(0)/2:=\kappa$. Let's now consider a covering of the unit ball with a set of cubes $\{C_{i}\}_{i=1,..,M}$ of diameter smaller than $\delta$. Then the subsets $S_{i}:=(C_{i}\times G_{d}(E))\cap \text{Supp}(\mu)$ form a partition of $\text{Supp}(\mu)$. Moreover, 
 \begin{align*}
  \|\mu\|_{W'}^2 & = \iint_{(E\times G_{d}(E))^{2}} k_{e}(x,y)\,k_{t}(u,v) d\mu(x,u) d\mu(y,v) \\
  & \geq \sum_{i=1}^{M} \iint_{S_{i} \times S_{i}} k_{e}(x,y)\,k_{t}(u,v) d\mu(x,u) d\mu(y,v) \\
  & \geq \kappa \sum_{i=1}^{M} \iint_{S_{i} \times S_{i}} k_{t}(u,v) d\mu(x,u) d\mu(y,v)\,.
 \end{align*}
We can apply the result of step 1 to each of these integrals, thus obtaining
 \begin{align*}
  \|\mu\|_{W'}^2 & \geq \beta \kappa \sum_{i=1}^{M} \mu(S_{i})^{2} \\
  & \geq \beta \kappa. \dfrac{1}{M} \left ( \sum_{i=1}^{M} \mu(S_{i}) \right )^{2} \\
  & \geq \dfrac{\beta \kappa}{M} \mu(E\times G_d{E})^{2}\,,
 \end{align*}
 the second inequality being a discrete Cauchy-Schwartz and the last one resulting from the fact that $\{S_{i}\}$ is a partition of $\text{Supp}(\mu)$. This ends the proof because the constants $\beta$ and $\kappa$ are both independent of $\mu$, and so is $M$ which only depends on the kernel $k_{e}$. 
\end{proof}
\noindent The hypotheses of theorem \ref{theo:principal} on the kernels are mostly technical but are not particularly restrictive in practice since one can check easily that all the kernels of equations (\ref{eq:kernel_CB1}) and (\ref{eq:kernel_G1}) comply to the requirements on $k_{t}$. As a direct corollary, we also have the following result ~:
\begin{cor}
 Let $X$ be a rectifiable subset of $E$ included in the unit ball. We make the same hypotheses on the kernel as in theorem \ref{theo:principal}. Then there exists a constant $\Cte$ independent of $X$ such that
  \begin{equation*}
  \|\mu_{X}\|_{W'} \geq \Cte\mathcal{H}^{d}(X)\,.
  \end{equation*} 
  \label{cor:rect_var}
\end{cor}

\begin{proof}
 This is essentially a special case of theorem \ref{theo:principal} for rectifiable varifolds. Indeed, the support of $\mu_{X}$ is included in $B(0,1) \times G_{d}(E)$ and thanks to the theorem, $\|\mu_{X}\|_{W'} \geq \Cte\mu_{X}(E \times G_{d}(E))$. In addition, from the very definition of $\mu_{X}$, $\mu_{X}(E \times G_{d}(E))=\mathcal{H}^{d}(X)$.
\end{proof}
\noindent This result is theoretically essential because it shows that pathological cases as the one of figure \ref{CEX_currents} cannot happen when shapes are represented as varifolds. In practical applications, it ensures the consistency of the kernel norm we use with the actual volume of the shapes so that artificial elimination of mass during registration process or template estimation are not likely to occur in this setting.  Note that when $X$ is of codimension one and is defined as the boundary of a bounded domain, the notion of volume of a shape $X$ is the area of the boundary of the domain and not the volume of the domain. In particular, the corollary says that when a plain shape is represented by its boundary, its norm for generic varifold kernel is larger than the area of its boundary.  We will show more specific examples of this in the next section.

A last fundamental issue to mention is the question of the variation of kernel norms with respect to the geometrical support of shapes, since registration and template estimation algorithms rely on the computations of gradients of such norms. We shall detail the technical computations of such gradients for discrete shapes in appendix B but it is also very valuable to have a theoretical and general interpretation of how the metric on shape varies when the support is deformed. Namely, if $X$ is a shape and a RKHS $W$ on varifolds is set, we want to express the variation with respect to $X$ of terms like $\langle \mu_{X}, \mu' \rangle_{W'}$ with $\mu'$ a certain element of $W'$ or equivalently of $\mu_{X}(\omega)$ for $\omega \in W$. The proper formulation is to express the derivative of $\mu_{X}(\omega)$ for small variations of $X$ obtained by flowing a vector field from $X$. With some hypotheses on $X$ and $W$, it is possible to derive a formula that generalizes in our context the notion \textit{first variation of a varifold} studied in \cite{Allard}. The result is summed up below~:
\begin{theo}
 \label{theo:variation_formula}
Let $X$ be an orientable compact submanifold and $\mu_{X}$ its associated varifold.  Let $v$ any $C^{1}$ vector field with compact support defined on $E$ and consider the associated one-dimensional subgroup $t\mapsto\phi_t$ generated by the flow of $v$. Then, if $X_t\doteq \phi_t(X)$ is the transported manifold, we have for any $C^{1}$ function $(x,V)\mapsto \omega(x,V)$ on $E \times G_{d}(E)$~:
  \begin{equation*}
   \left. \dfrac{d}{dt} \right \vert_{t=0}\mu_{X_t}(\omega)= \int_{X} \left ( \dfrac{\partial \omega}{\partial x} - \mdiv_{X} \left (\dfrac{\partial \omega}{\partial V} \right ) - \omega H_{X} | v^{\bot} \right ) + \int_{\partial X} \langle \nu, \omega v^{\top} + \left( \dfrac{\partial \omega}{\partial V} | v^\bot \right )\rangle
  \end{equation*}
where $v^\bot$ and $v^\top$ denote the tangential and normal part on $v$ along $X$ and $\nu$ is the unit outward normal along $\partial X$.
 \end{theo}
We refer to appendix B for the precise definition of all terms and for the complete proof of this formula. We will just make a few qualitative comments since there are a few noteworthy consequences to mention. The first important remark is that the variation of varifold metric is controlled only by the vector field $v$ and not its derivatives. This is not straightforward since a varifold kernel $K((x,V),(x',V'))$ is encoding in an arbitrary non-linear way first order information through the inclusion of Grassmannian variables $V$ and $V'$. Interestingly, the result is valid even if the dimension or the codimension of $X$ is different from $1$. In return, we see that the formula involves some terms on the \textbf{boundary} of $X$, one of them expressing the tangential extension of $X$ along its boundary and a second one related to the variation in the tangent space direction on the boundary in the normal direction $v^\bot$. In the interior of $X$, we see that the variation depends only on the orthogonal component $v^{\bot}$ of $v$, which, in other terms, shows that the gradient of the attachment distance is orthogonal to the shape.        

\section{Large deformation matching of unoriented shapes~: an algorithm}
\label{sec:var_LDDMM}
At this point, we have defined a theoretical background to represent and compare unoriented shapes through varifolds and kernels on varifolds. What we have obtained is a distance (actually a whole class of distances provided by different kernels) between the objects that can serve as an attachment term in practically any matching process. In the rest of the article, we will focus on one particular model of large deformations called \textbf{LDDMM}, which has already proved its interest for various sorts of data, from images \cite{Trouve2}, curves and surfaces \cite{Glaunes2} \cite{Durrlemann3} to fiber bundles \cite{Durrlemann4} and more recently functional shapes \cite{Charon1}. The purpose of the following is to derive the equations needed for the numerical implementation of LDDMM on varifolds. In all this section, we will restrict ourselves to the usual cases of curves and surfaces living in the 3-dimensional euclidean space $\mathbb{R}^{3}$. 

\subsection{Description of the algorithm}
\label{sec:var_LDDMM_algo}
Let's briefly remind the principal features of LDDMM modelling. Given two objects $\mathcal{O}_{1}$ (the source) and $\mathcal{O}_{2}$ (the target) of the ambient space $E$ (two curves, two surfaces...), the registration problem consists in finding an 'optimal deformation' that transforms the source object onto the target one, or at least approximately. In the LDDMM framework, we consider a group of deformations that are diffeomorphisms given as the flow of time-varying vector fields of the space. If $V$ is a certain Hilbert space of vector fields on $E$ and $L^{2}_{V}([0,1])$ denotes the space of time-varying vector fields  $v_{t}$ with $\forall t \in [0,1], \ v_{t}\in V$ and $\int_{0}^{1} |v_{t}|_{V}^{2} dt < \infty$, then the group of deformations associated with $V$ denotes $G_V$ is the set of all $\phi_{1}^{v}$ satisfying $\phi_{0}^{v}=Id$ and the differential equation 
$$\frac{\partial \phi_{t}^{v}}{\partial t} = v_{t}\circ \phi_{t}^{v}\,.$$
The underlying geometrical setting is to consider $G_V$ as equipped with a right invariant metric induced by $V$ playing the role of the tangent space at identity \cite{trouve98:_diffeom}. Matching the two objects can be formulated as finding the solution of the variational problem~:
 \begin{equation}
\label{eq:LDDMM_general}
 \inf_{v \in L^{2}_{V}([0,1])} J(v), \ J(v)=\frac{1}{2}\int_{0}^{1} |v_{t}|_{V}^{2} dt + \frac{\gamma}{2} d(\phi_{1}^{v}.\mathcal{O}_{1},\mathcal{O}_{2})^2 
 \end{equation} 
This functional $J$ can be interpreted as the sum of two terms, the first one that constrains the square length of the deformation path (which is the meaning of the word optimal deformation used above) while the second term drives the matching by measuring the distance of the deformed source $\phi_{1}^{v}.\mathcal{O}_{1}$ to the target. $\gamma$ is a trade-off parameter between the two terms. On the optimality of the deformations due to the minimization of the first term, much has been said and proved in the past, notably with the interpretation as geodesics in shape spaces (cf \cite{Younes}). This resulted in two main algorithms for the resolution of the variational problem, namely a \textit{gradient descent} scheme in the space $L^{2}_{V}([0,1])$ and later a \textit{geodesic shooting} algorithm. Both these schemes are essentially related to the dynamics of the deformations and not to the nature of the objects and can be therefore adapted to the varifold case in a similar way. To fix ideas, in the following, we shall focus on the gradient descent algorithm for non-oriented shapes but the equations given for the data attachment distance and its gradient could be plugged into a geodesic shooting procedure almost straightforwardly.  \\
If the source object $\mathcal{O}_{1}$ is given as a finite set of points in the ambient space $\{q^{i}\}_{i=1,..,N_{1}}$ and if the vector field's space $V$ is taken as a RKHS with kernel $K_{V}$, then it has been shown in \cite{Glaunes} that the optimal vector field in the problem (\ref{eq:LDDMM_general}) can be searched under the particular form~: 
 \begin{eqnarray}
\label{eq:LDDMM_vecfieldform}
 v_{t}(x)&= &  \sum_{i=1}^{N_{1}} K_{V}(q_{t}^{i},x)\alpha_{t}^{i}\,, \\
  \text{ with } & & q_{t}^{i}=q^{i}+\int_{0}^{t} v_{s}(q_{s}^{i}) ds\,. \nonumber
 \end{eqnarray} 
This means that the optimal vector field at all times $t$ is completely parametrized by the set of vectors $\alpha_{i}^{t} \in E$ that we call the \textit{momenta} of the deformation. The optimization of $J$ is therefore equivalent to the minimization of the following functional with respect to $\alpha=(\alpha_{t}^{i})$~:
\begin{equation}
\label{eq:LDDMM_alpha}
\left\{
  \begin{array}[h]{l}
  \tilde{J}(\alpha)=\frac{1}{2}\int_{0}^{1} \sum_{ij}\langle K_{V}(q_{t}^i,q_{t}^i)\alpha_{t}^j,\alpha^i_t\rangle dt + g(q_{1}) \\
\\
 \dot{q}_{t}^i=\sum_{j}K_{V}(q_{t}^i,q_{t}^j)\alpha_{t}^j\,.
  \end{array}\right.
\end{equation}
This leads to finite dimensional optimal control problem where the vector $\alpha_t$ is the control variable. As detailed in \cite{Trouve2011}, the variation of $\tilde{J}$ with respect to a variation $\delta \alpha$ of $\alpha$ gives:
$$\delta \tilde{J}(\alpha)=\int_0^1\sum_{i}\langle \sum_{j}K_V(q^i_t,q^j_t)(\alpha^j_t-\eta^j_t),\delta \alpha^i_t\rangle dt$$
where the covariable $\eta$ is solving 
the backward integration scheme~:
\begin{equation}
\label{eq:LDDMM_backward}
\left\{
  \begin{array}[h]{l}
  \dot{\eta}_{t}=-dv^{\ast}(q)\eta_{t} \\
\\
 \eta_{1}+\nabla g(q_{1})=0
  \end{array}\right.
\end{equation} 
where $dv^*$ is the adjoint of $dv$ and $v$ is given by (\ref{eq:LDDMM_vecfieldform}). Since $(K_V(q^i_t,q^j_t)_{ij})$ is positive definite, $t\mapsto \eta_t-\alpha_t$ as a descent direction for $\tilde{J}$.

At this point, we must insist on the fact that the differential equations of (\ref{eq:LDDMM_alpha}) and (\ref{eq:LDDMM_backward}) are focused on the dynamics of optimal diffeomorphisms, independently from the nature of the objects, whether they be sets of landmarks, oriented or unoriented curves and surfaces, measures... The only requirement is to specify the attachment quantity $g(q_{1})$ for the computation of the total energy and its gradient $\nabla g(q_{1})$ for the initialization of $\eta$ in the backward equation. Adapting LDDMM to unoriented sets represented as varifolds therefore consists essentially in computing these terms. We refer to appendix B where we have detailed such computations for the case of embedded curves and surfaces in $\mathbb{R}^3$ compared through a kernel metric on the space of varifolds. The rest of the algorithm, similarly to LDDMM for landmarks or currents, is a gradient descent on the $\alpha_{t}$ computed through equation (\ref{eq:LDDMM_backward}).

\subsection{Some simulations and results}
We now come to a few results of varifold LDDMM algorithm, in which we want to emphasize the benefit in situations that traditionally involve orientation issues when objects are matched with the currents' framework. In all the following experiments, we chose the Gaussian kernel \\ 
$k_{e}(x,y)=e^{-\frac{|x-y|^{2}}{\sigma_{e}^{2}}}$ on the space $E$. As discussed previously, there are also many possibilities for the kernel $k_{t}$ on the Grassmann manifold~: we have focused essentially on the Cauchy-Binet kernel and the Gaussian kernel. Even though both kernels were proved to induce a distance on the set of reunion of submanifolds, we have also explained why the Gaussian kernel has better separation properties, especially when the scale $\sigma_{e}$ of the kernel on $E$ is chosen to be large. Still, the Cauchy-Binet kernel has the advantage of not introducing an additional scale parameter and also leads to fast numerical computations for shapes with high number of points. We have therefore used the Gaussian kernel with experiments on curves but rather the Cauchy-Binet kernel for 3D surfaces which are usually more sampled. \\
The example of figure \ref{matching_ninja_star} shows a situation with pikes' structures that are typically not well matched using LDDMM with currents. This example can be also used to show the robustness of varifolds when we increase the scale parameter of the spatial kernel $k_{e}$, as illustrated in figure \ref{matching_ninja_scale}. It is particularly remarkable that even at scales much larger than the size of the object, the algorithm is still able to capture the overall description of the target shape, simply based on the distribution of tangent spaces (cf discussion at the beginning of section \ref{sec:kernel_var_properties}). 
\begin{figure}
\begin{tabular}{cc}
\includegraphics[width=5cm,height=5cm]{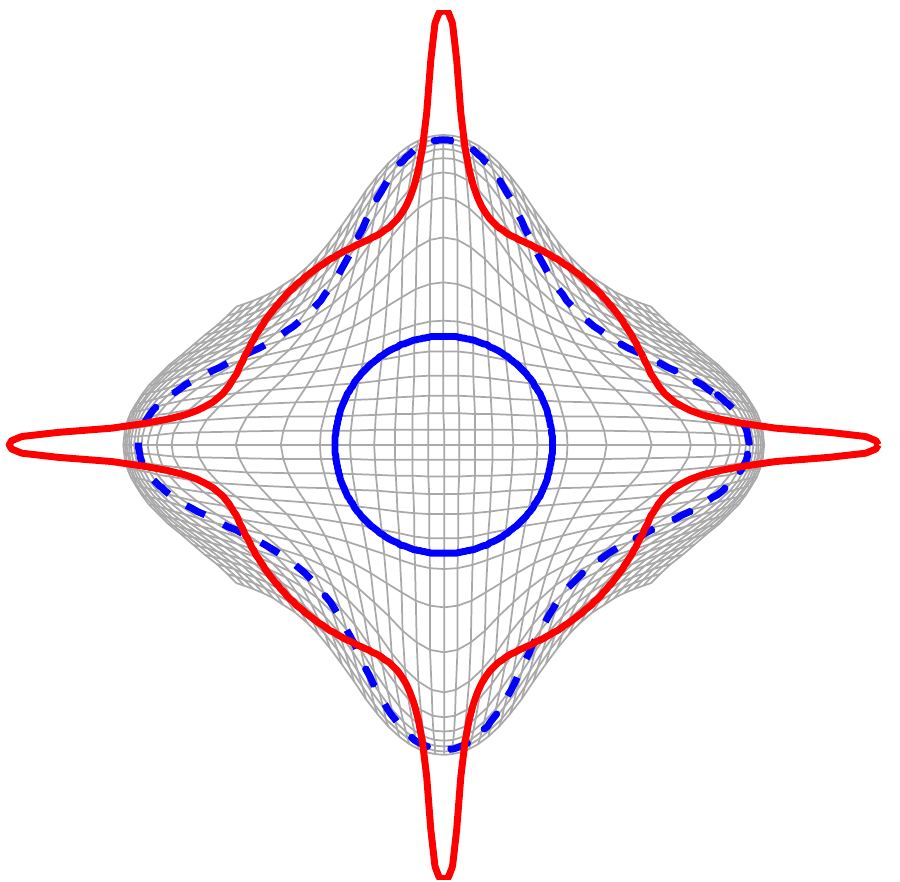} & \includegraphics[width=5cm,height=5cm]{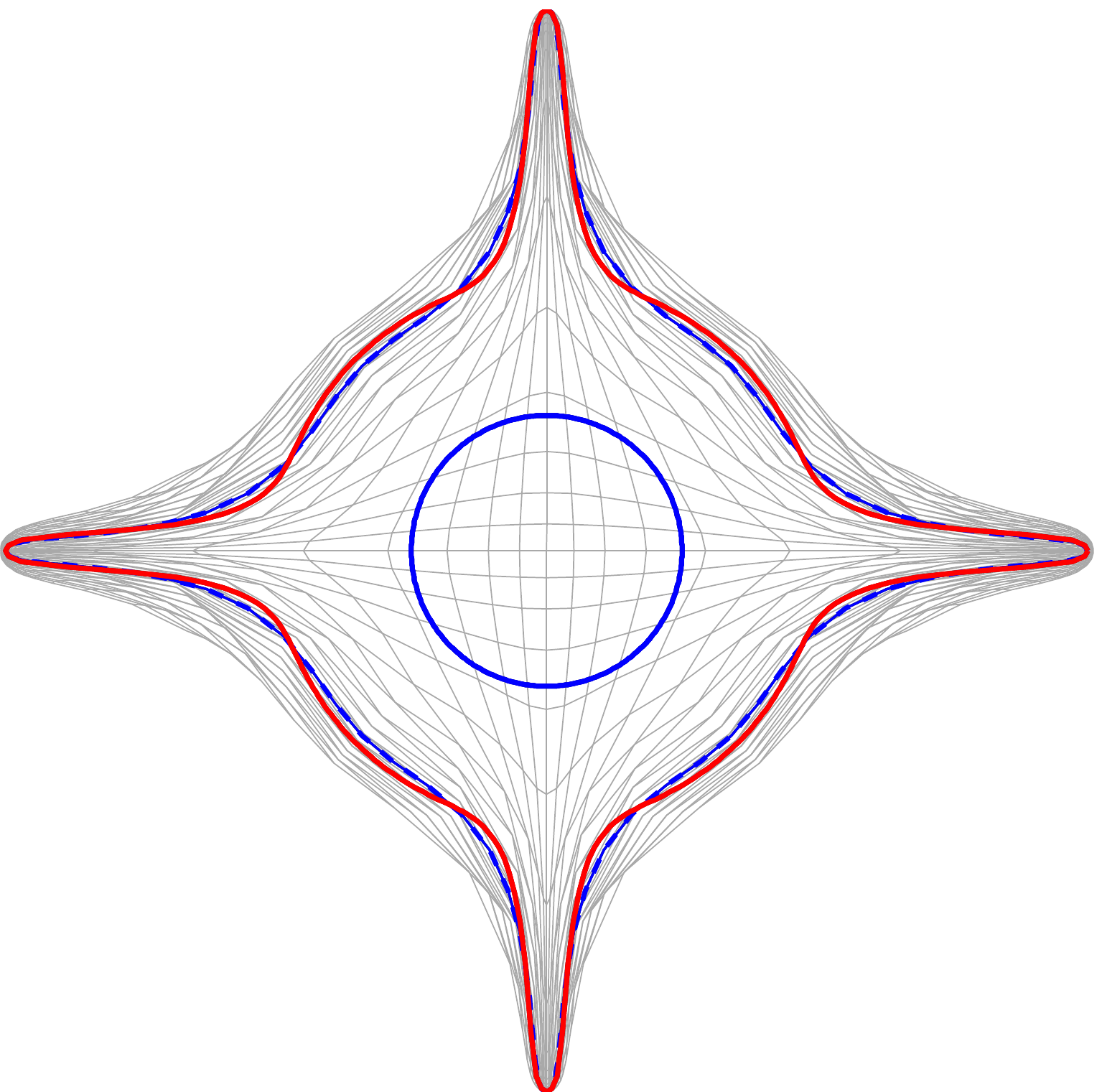}
\end{tabular}
\caption{Example of registration between 2D curves. The source shape is the blue circle, the target is the red star with four narrow branches. On the left, the matching is performed with the approach of currents. On the right, with the approach of varifolds exposed in section \ref{sec:var_LDDMM_algo} with the same parameters. We see that the branches are well recovered with varifolds whereas the current's metric is nearly insensitive to them.}
\label{matching_ninja_star}
\end{figure}

\begin{figure}
\begin{tabular}{cc}
\includegraphics[width=5cm,height=4cm]{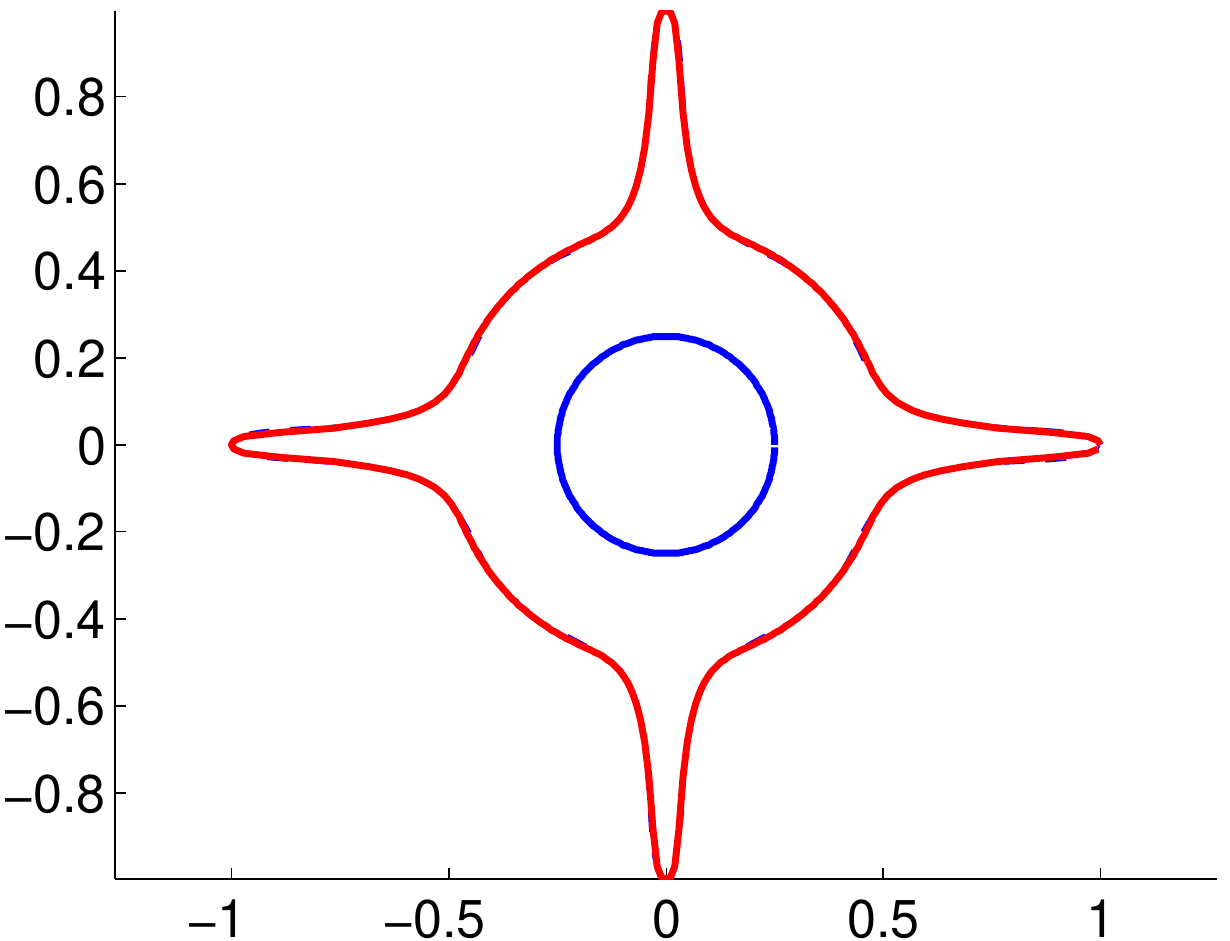} & \includegraphics[width=5cm,height=4cm]{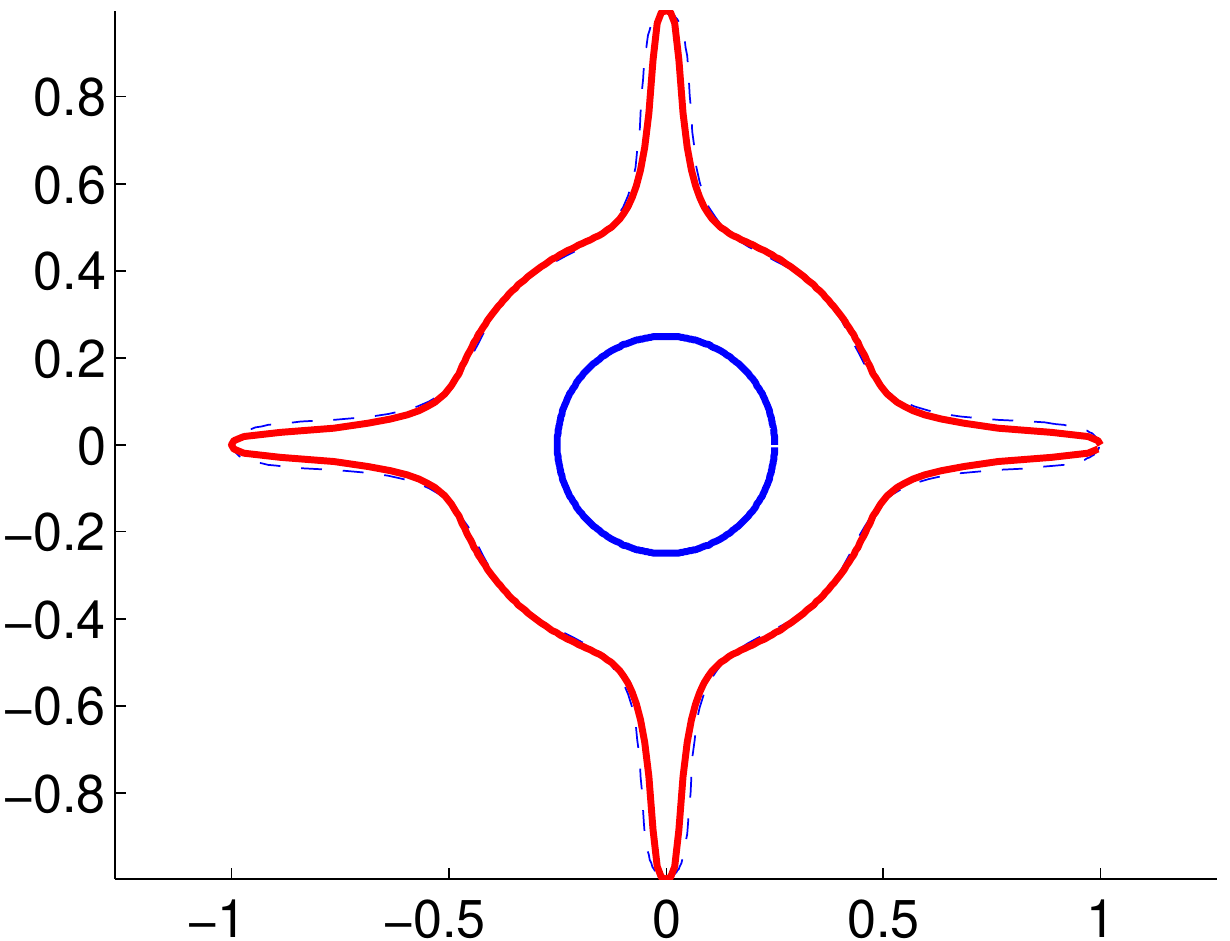} \\
$\sigma_{e}=0.5$ & $\sigma_{e}=1$ \\
\includegraphics[width=5cm,height=4cm]{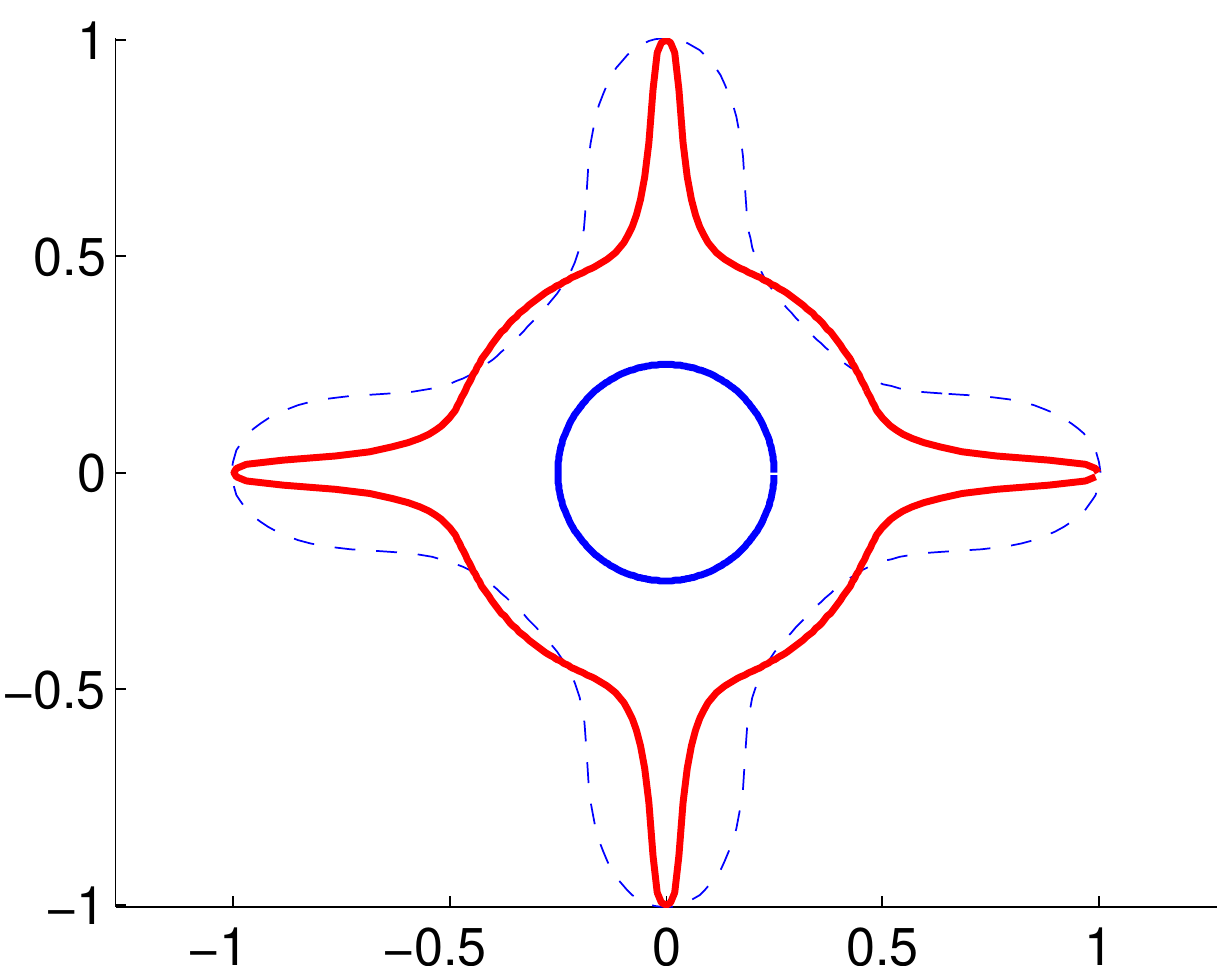} & \includegraphics[width=5cm,height=4cm]{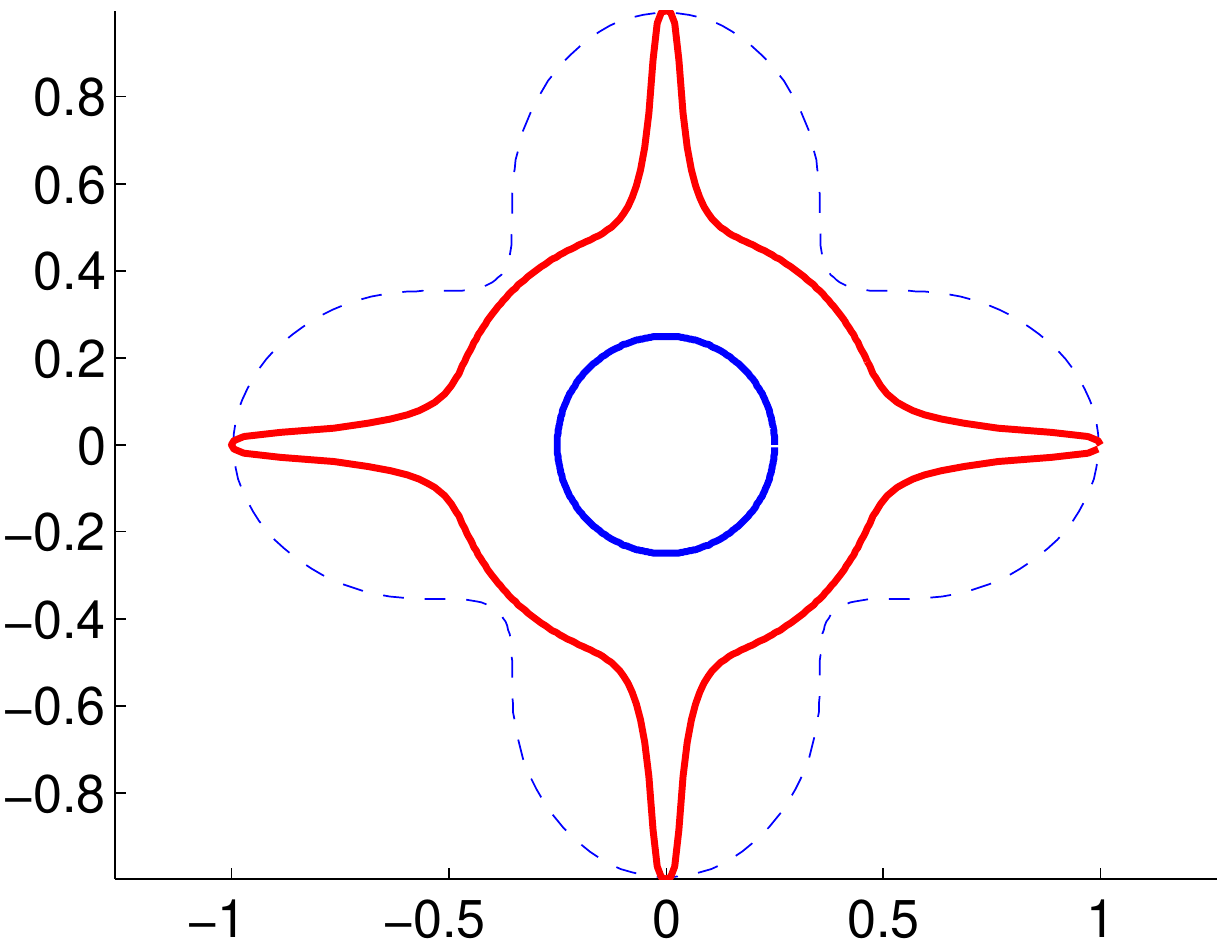} \\
$\sigma_{e}=1.5$ & $\sigma_{e}=10$
\end{tabular}
\caption{The same example for varifolds but performed at different scales $\sigma_{e}$.}
\label{matching_ninja_scale}
\end{figure}

In figure \ref{matching_misoriented_fibres_var}, we go back to the example of figure \ref{matching_misoriented_fibres} where shapes are constituted of many disconnected pieces of curves which are not consistently oriented, and show the results of both methods (currents and varifolds). Other interesting situations are given by tree-like structures for which orientating the different branches consistently for registration can become a nearly intractable problem when faced with many intersecting branches. We give a simple insight of this phenomenon in figure \ref{matching_trees}. \\
We also present some results of surface matching done with our varifold algorithm. The first one is a matching between two thin envelopes. Due to the proximity of the lower and upper membranes, current-norm based registrations have a tendency to squeeze both membranes together in order to eliminate unmatched parts of the shapes, which we can see on figure \ref{matching_enveloppes}. Finally, in figure \ref{matching_bunny}, another surface registration is applied between a standard sphere and the Stanford bunny model. We observe again some mismatch artifacts in the currents's result~: the ears are not fully recovered as well as some small details in the head and we see the apparition of undesirable membranes at the basis of the ears. In contrast, the varifold algorithm matches almost perfectly all parts of the shapes.

\begin{figure}
\leftskip -1cm
\begin{tabular}{cc}
\includegraphics[width=7cm,height=7cm]{matching_misoriented_fibres_current.pdf} & \includegraphics[width=7cm,height=7cm]{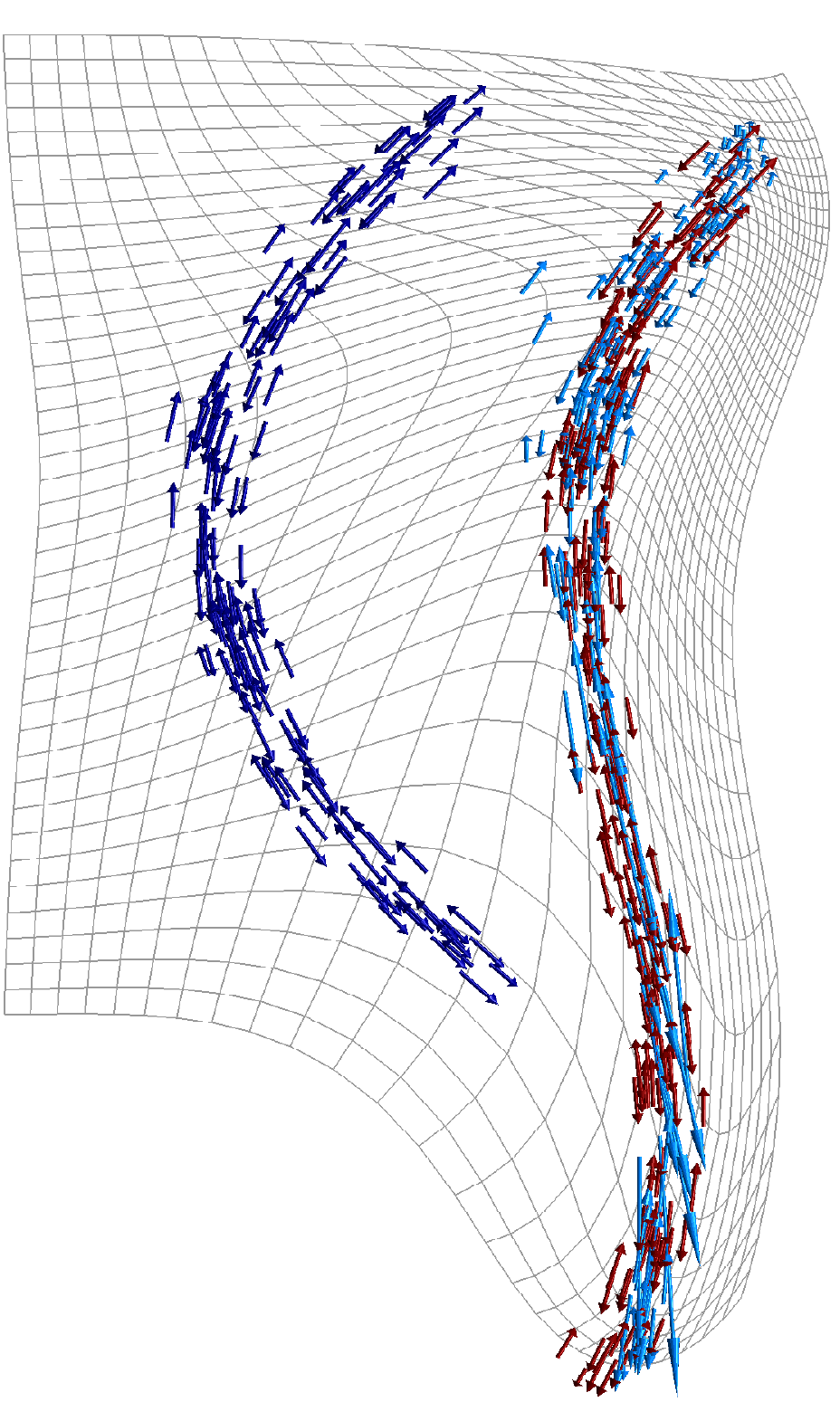}
\end{tabular}
\caption{Registration result on the example of figure \ref{matching_misoriented_fibres} with the framework of currents (left) and of varifolds (right). As expected, varifolds can deal efficiently with any orientation of the different components.}
\label{matching_misoriented_fibres_var}
\end{figure}

\begin{figure}
\leftskip -1cm
\begin{tabular}{cc}
 \includegraphics[width=7cm,height=7cm]{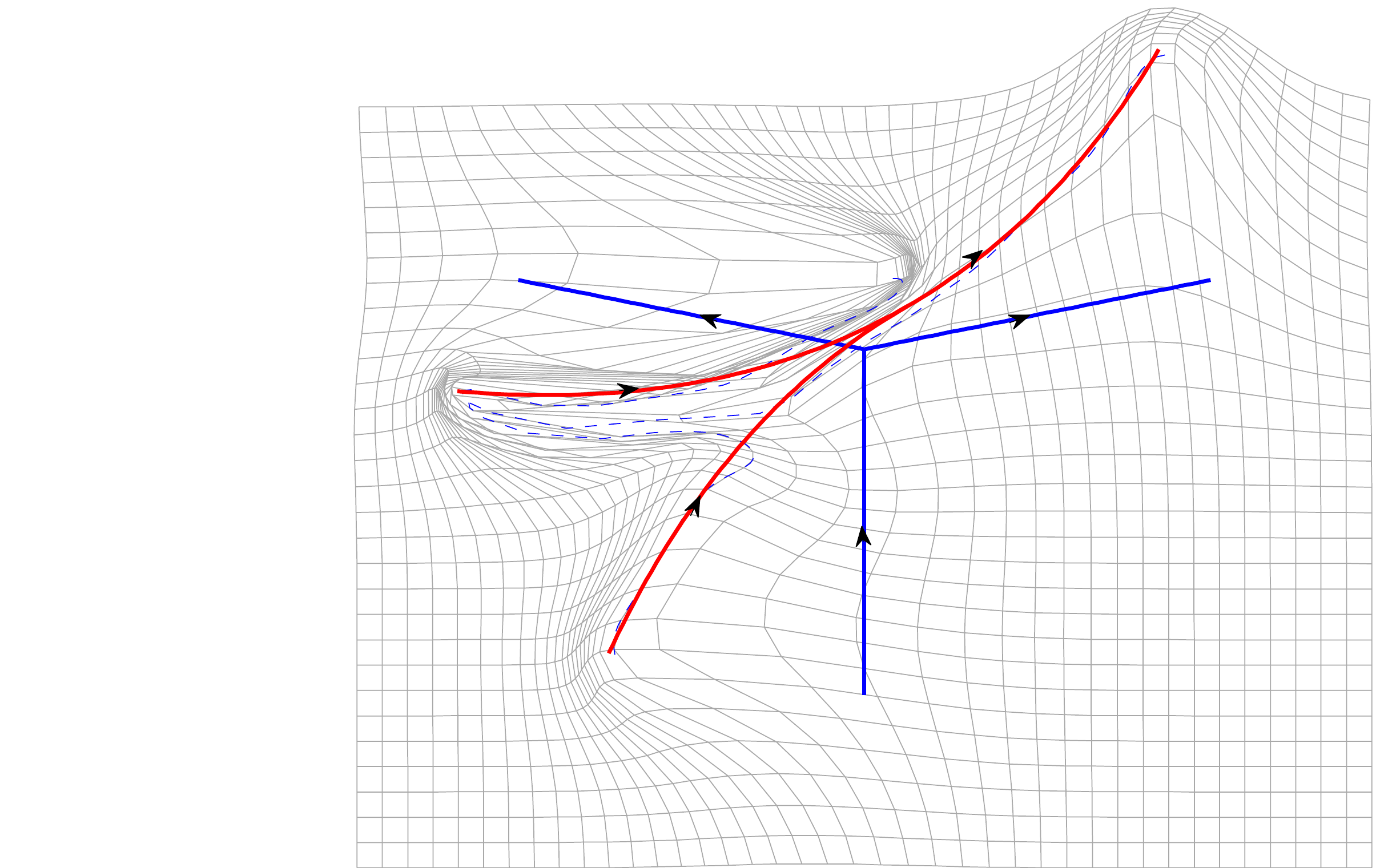} & \includegraphics[width=5.5cm,height=6.5cm]{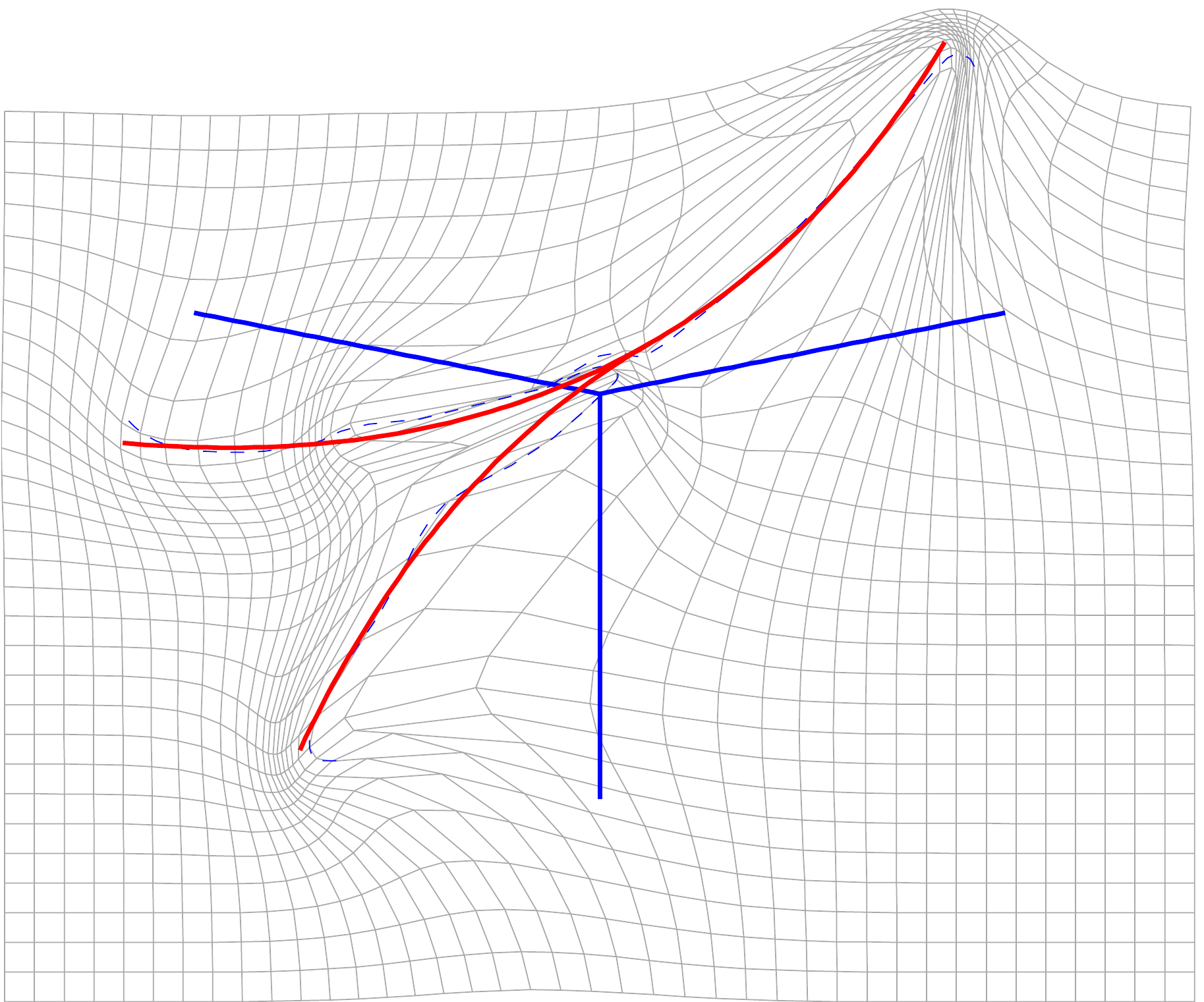}
\end{tabular}
\caption{Example of registration on trees with currents (left) and varifold (right). Observe again that, with currents, consistent orientation between the source and the target is necessary to avoid the unnatural deformation shown on the figure.}
\label{matching_trees}
\end{figure}

\begin{figure}
\begin{tabular}{cc}
\includegraphics[width=5cm,height=5cm]{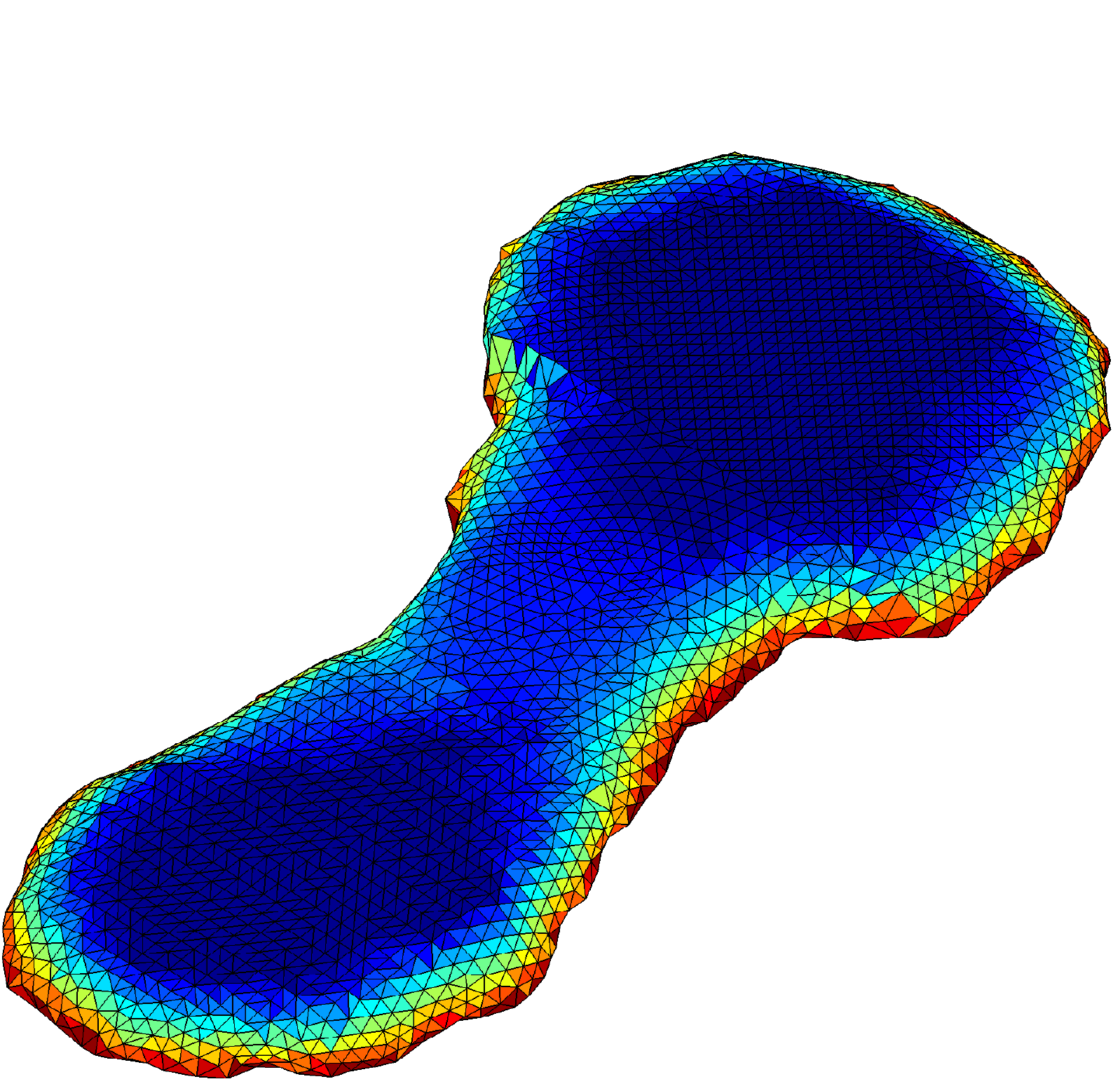} & \includegraphics[width=5cm,height=5cm]{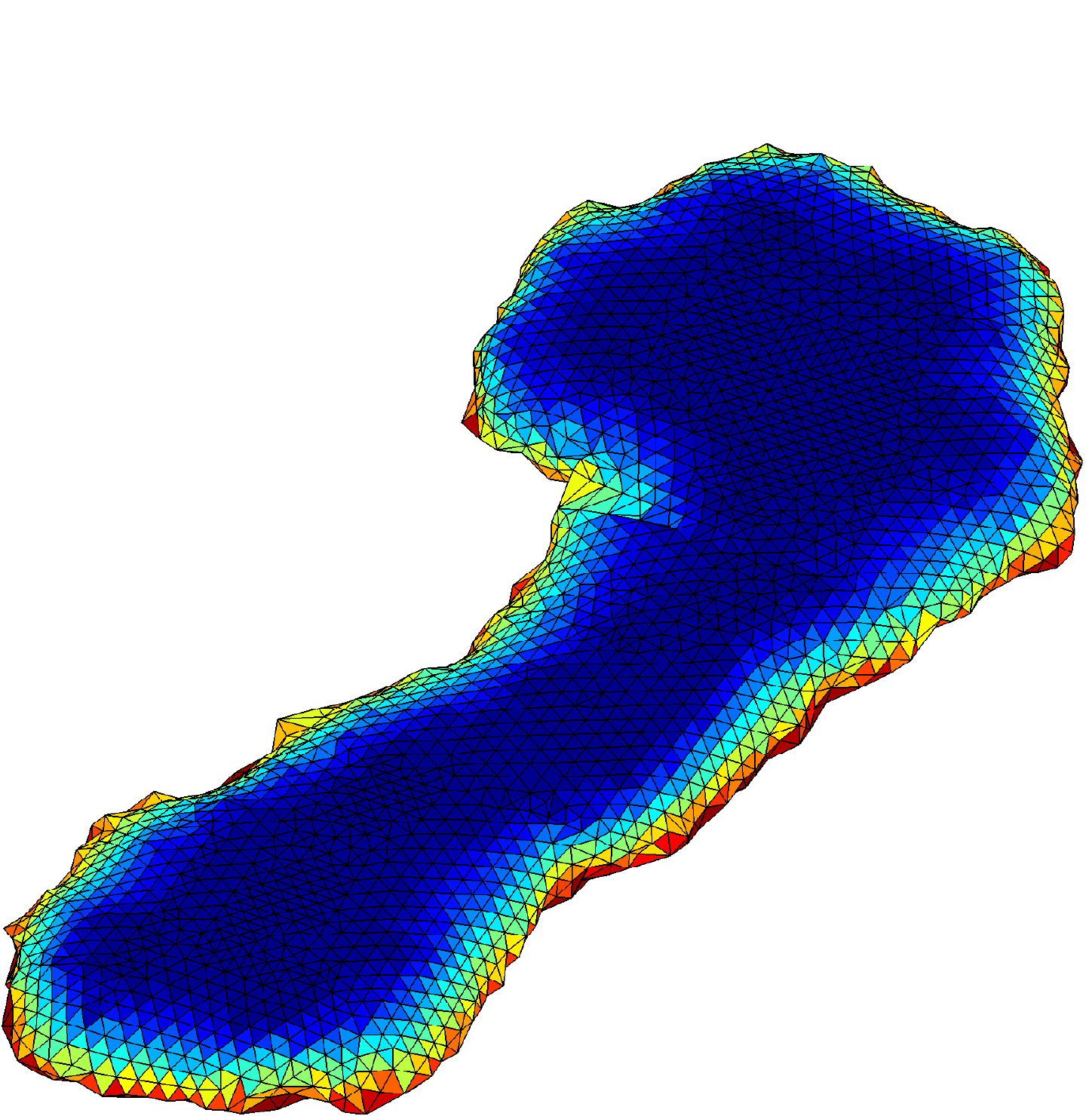} \\
Source surface & Target surface \\
\includegraphics[width=5cm,height=5cm]{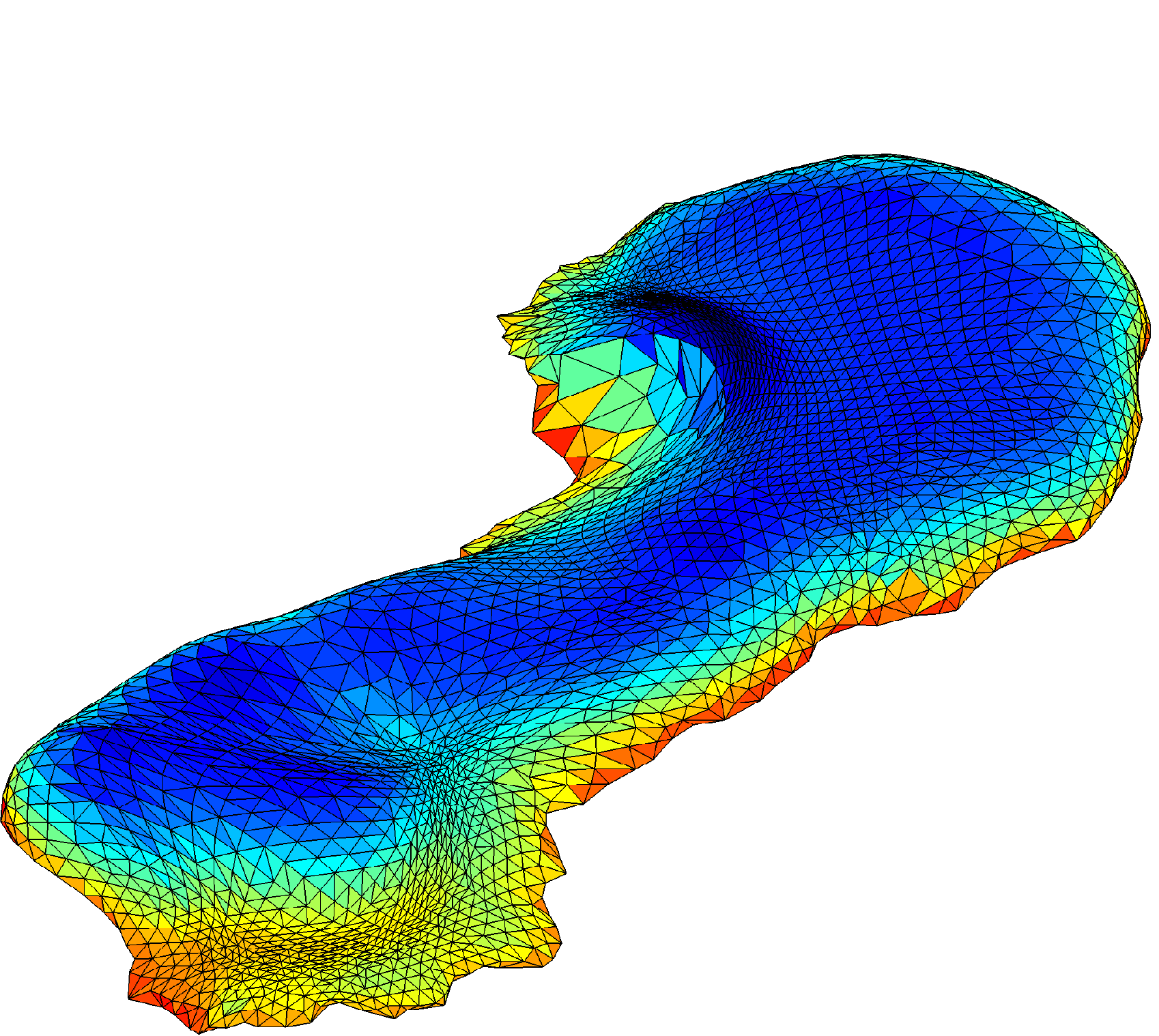} & \includegraphics[width=5cm,height=5cm]{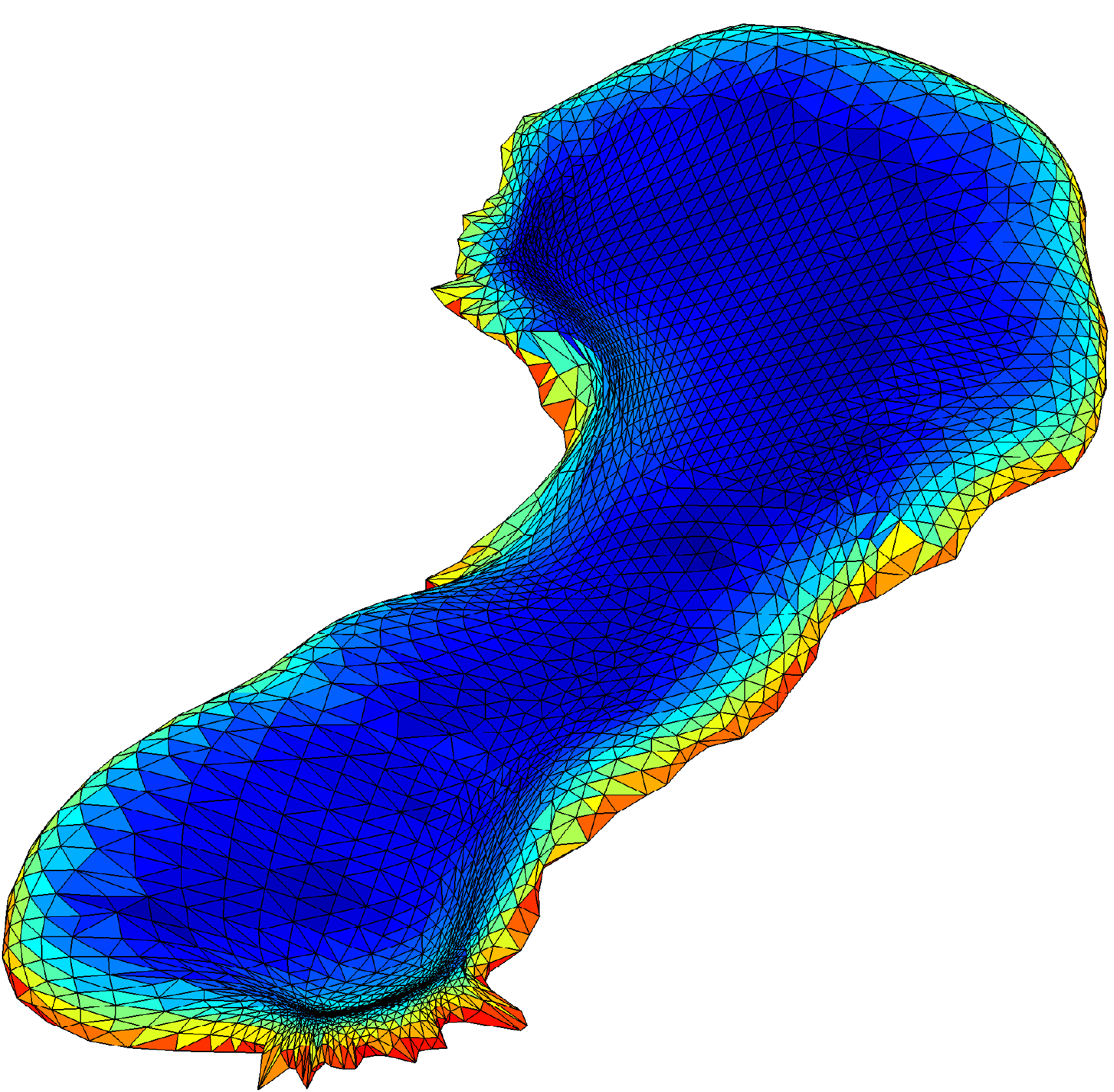} \\
Matching with currents & Matching with varifolds
\end{tabular}
\caption{Registration of two heart envelopes from mouse embryos (see Le-Garrec et al, \cite{LeGarrec}). We show the source and target envelopes on top, and the result after matching with currents and varifolds. Observe that the algorithm based on currents is able to reduce the data attachment distance between source and target by crushing both sides together instead of displacing them, unlike the varifold framework.}
\label{matching_enveloppes}
\end{figure}

\begin{figure}   
\begin{tabular}{cc}
\centering
 \includegraphics[width=5cm,height=5.5cm]{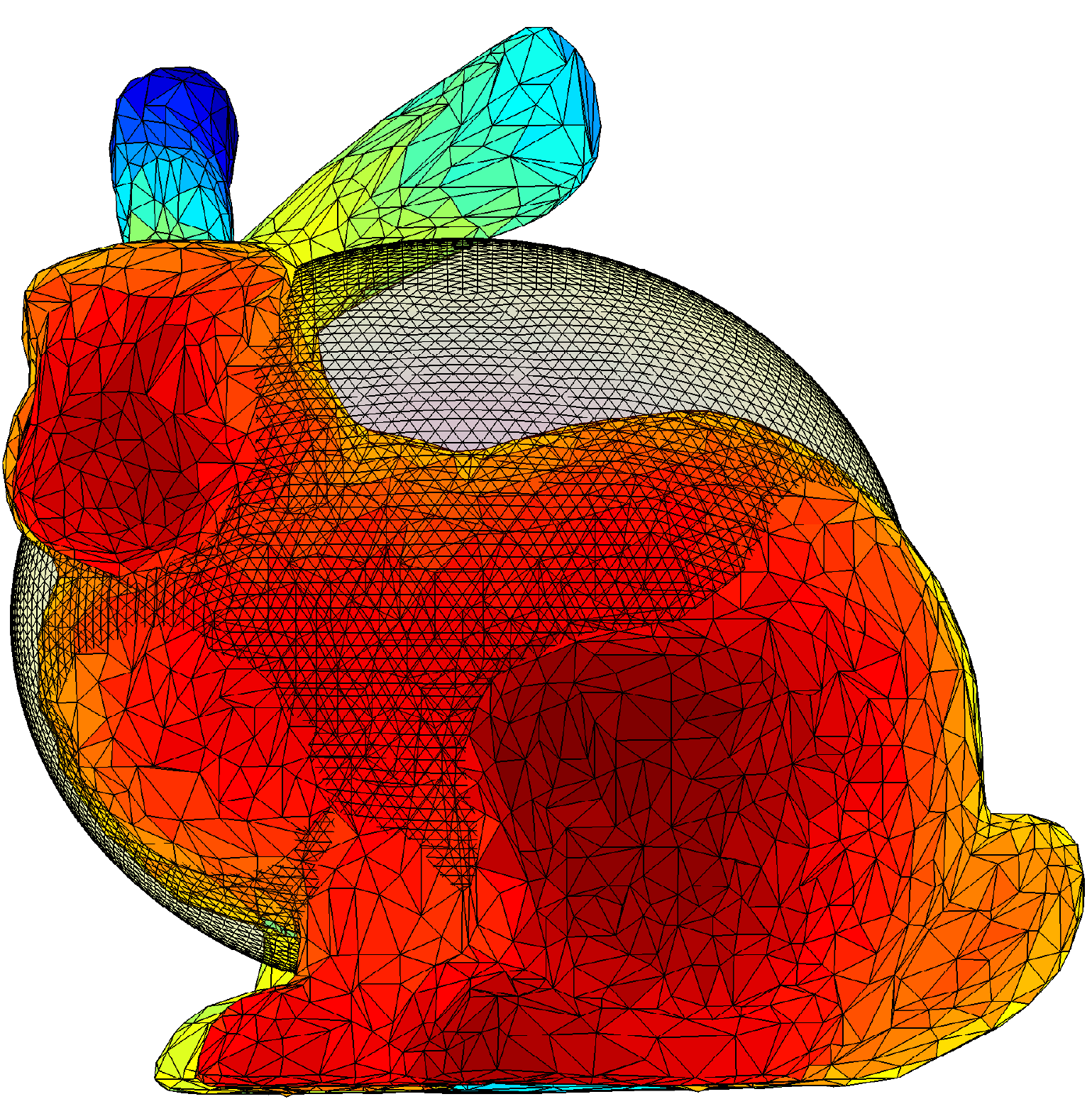} & \includegraphics[width=5cm,height=5.5cm]{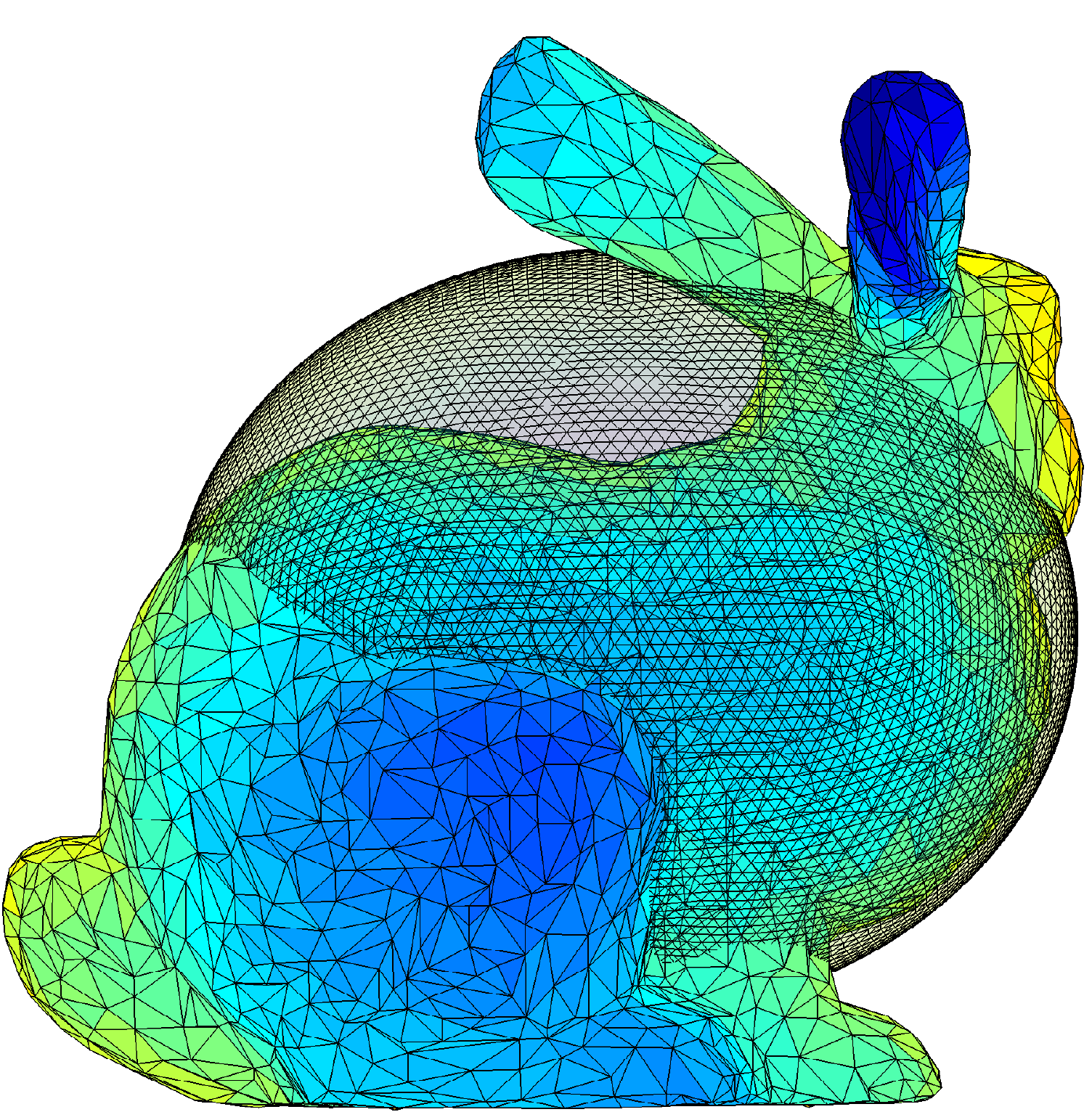} \\
 \includegraphics[width=5cm,height=5.5cm]{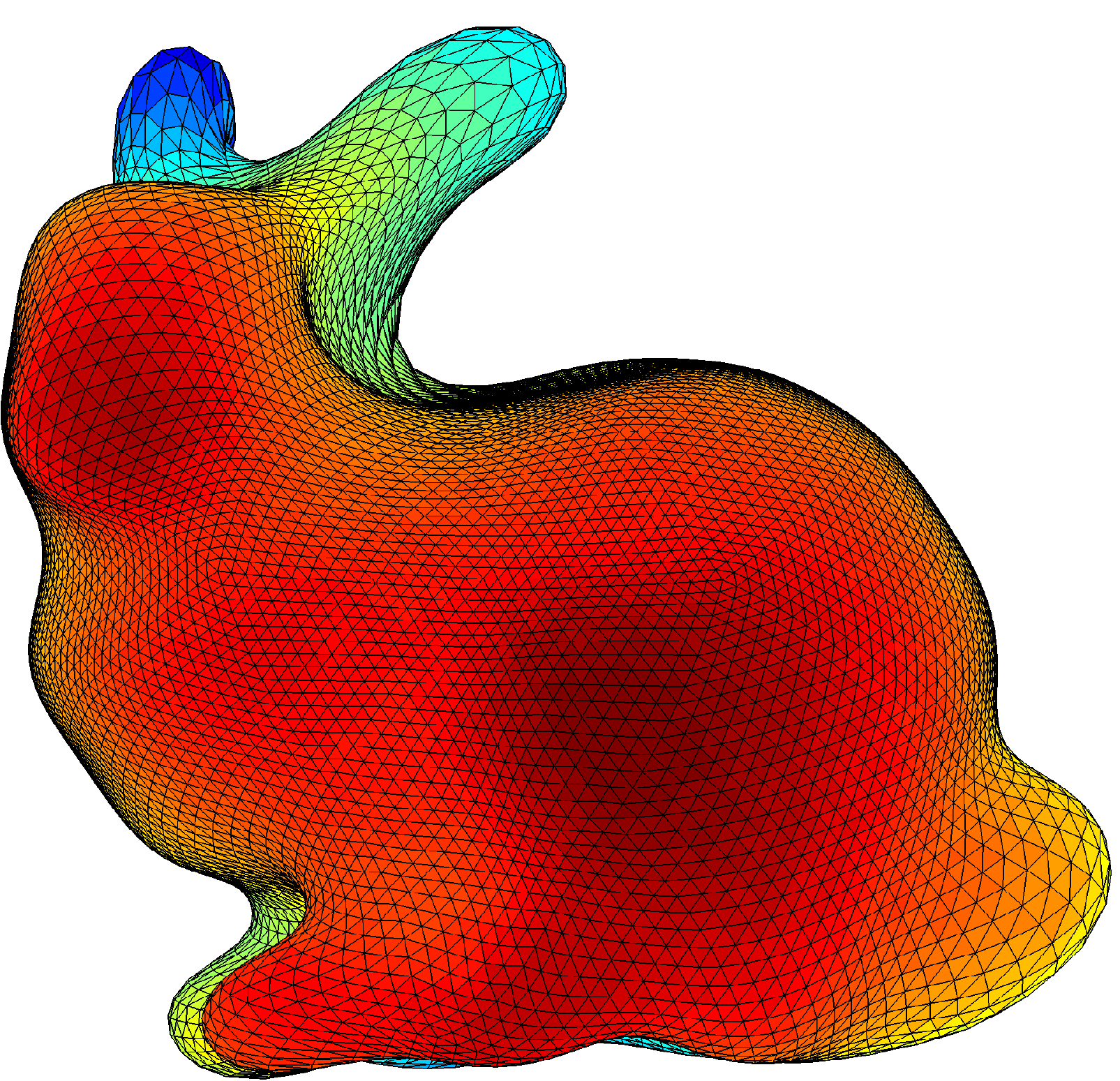} & \includegraphics[width=5cm,height=5.5cm]{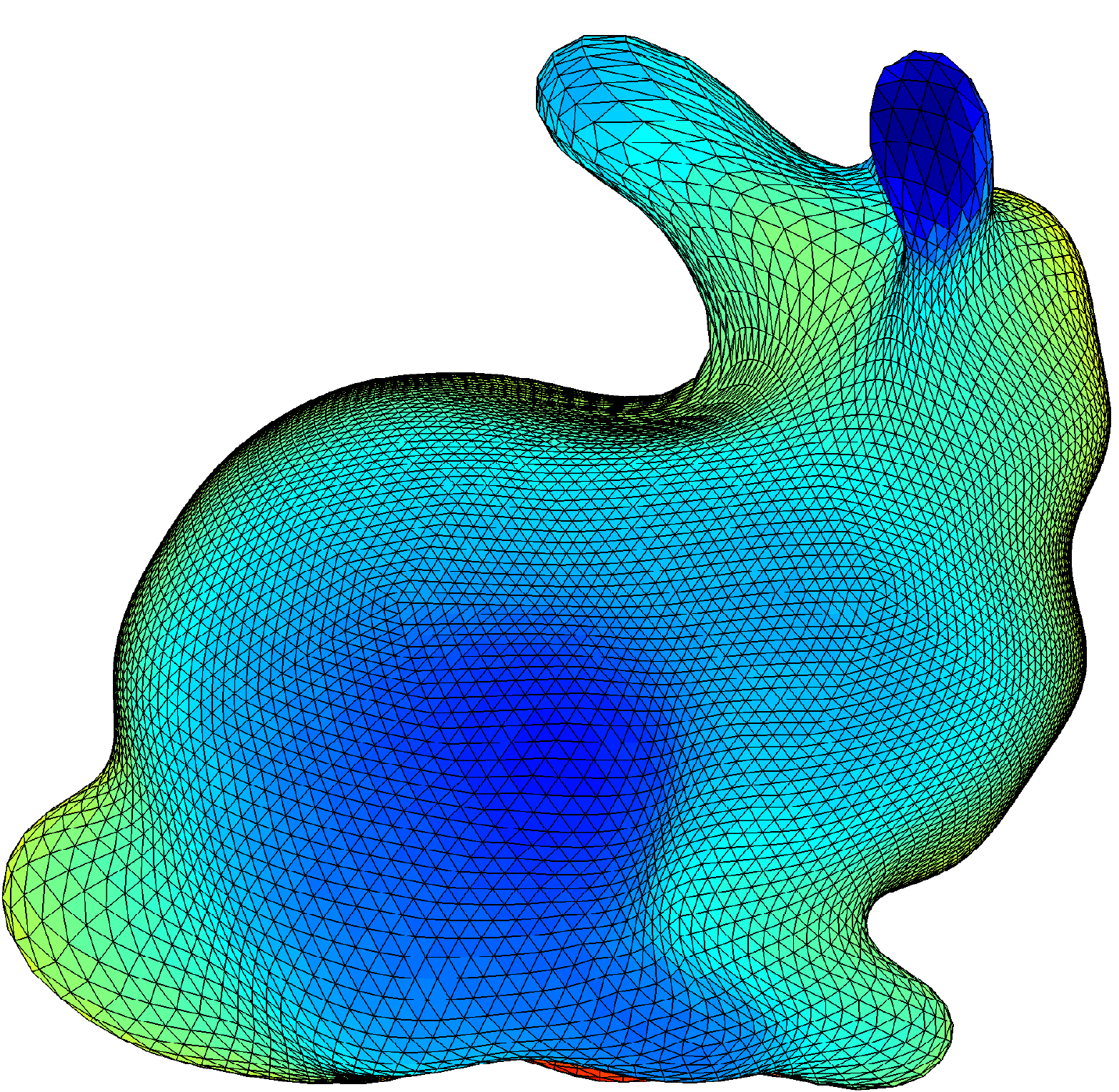} \\
 \includegraphics[width=5cm,height=5.5cm]{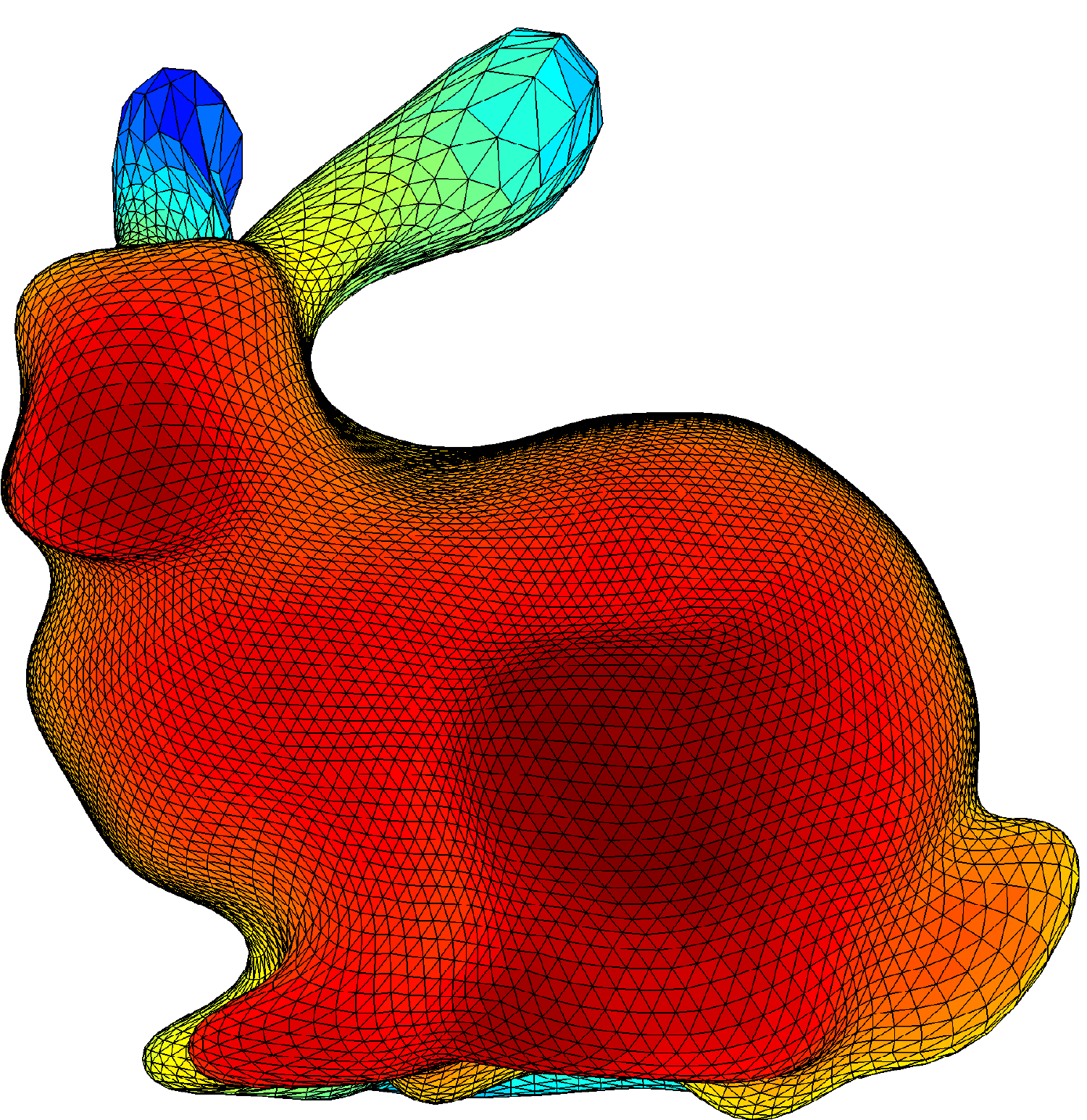} & \includegraphics[width=5cm,height=5.5cm]{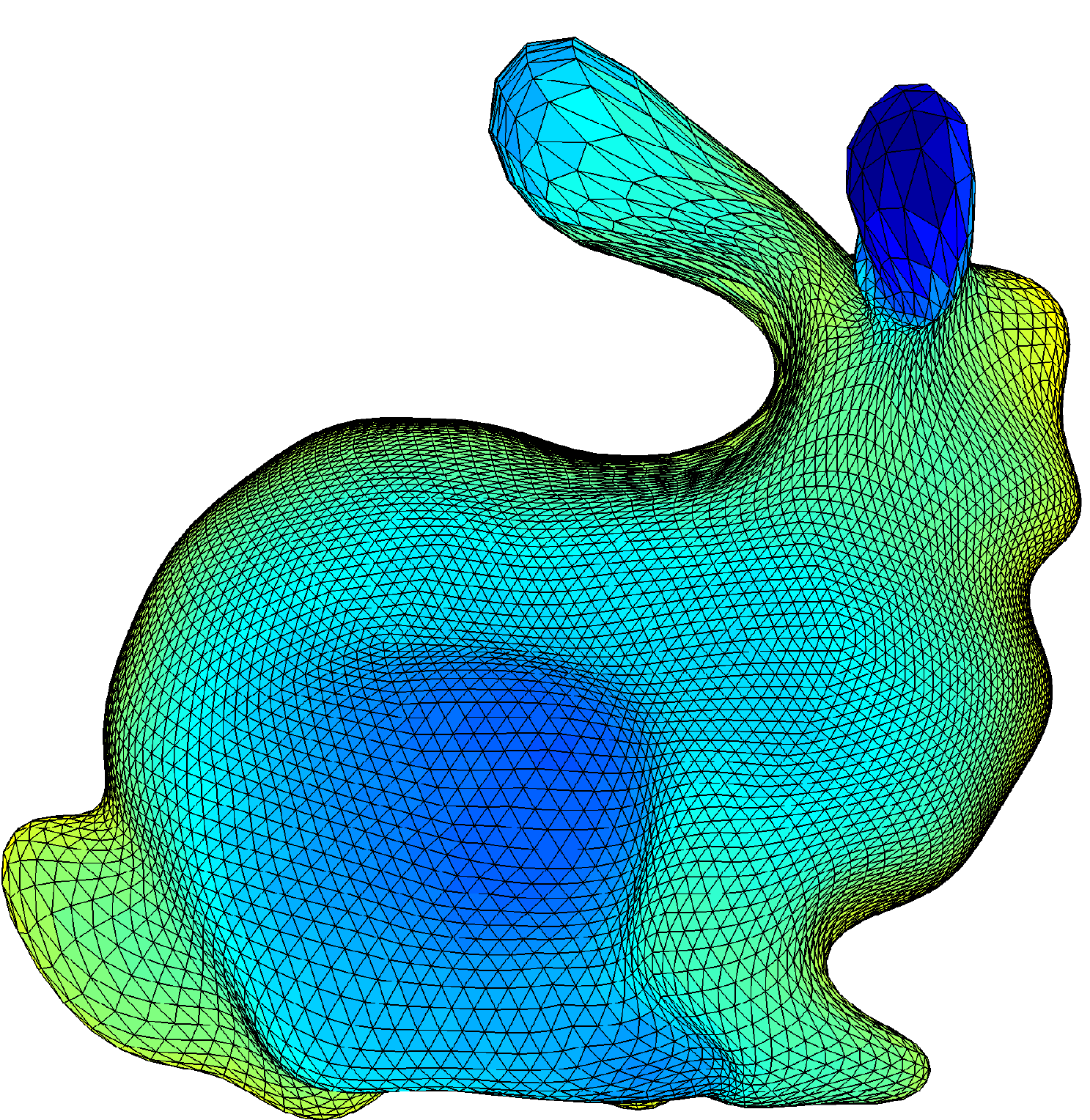}
\end{tabular}
\caption{Surface matching results between a sphere and the Stanford bunny. At top row, the source sphere and the target bunny shape. In the middle row, the result after matching with currents. At bottom row, with the varifold approach.} 
\label{matching_bunny}
\end{figure}

\section{Conclusion}
Having explored in the first place the orientation issues that come along with the use of currents in shape modelling, we have adapted in this paper the alternative concept of varifolds, by representing shapes as spatially spread distributions of unoriented tangent spaces themselves modeled as elements of the Grassmann manifold. As for currents, we embed a wide variety of objects like submanifolds and rectifiable subsets in one common functional space. The first contribution of our work was to define metrics that share theoretical requirements and enable convenient numerical computations. We have shown how this can be done by defining Hilbert structures on varifolds through reproducing kernels and proposed a general construction process for such kernels. From the theoretical side, we proved that these metrics, as opposed to the previous representation by currents, do not artificially 'eliminate' shape volume due to orientation. We also computed a variation formula of the metric with respect to shapes. In terms of numerical applications, we have presented an adaptation of large deformation shape registration based on this framework. The results we presented are mainly focused on synthetic examples for the time being, but we believe there are already important conclusions and perspectives that are brought out. The first one is the fact that our approach does not require any orientation of the shapes, which, in our opinion, can be of great interest in datasets which are not easily or even not orientable in a consistent way. Another point we would like to suggest is the idea that the tangent space distribution of a shape, even if the positions are not accurate (e.g for large scale kernels on $E$), is already giving much information on the shape itself. Because of orientation and the canceling effect, currents totally loose such information. In contrast, varifold-based algorithms are able to recover shapes, at least approximately, even at very large spatial scale. 

\section*{Acknowledgments}
Authors want to thank Stanley Durrleman and Marcel Prastawa for developing and providing the C++ code for shape registration used in the examples of this paper. This work was made possible thanks to HM-TC (Hippocampus, Memory and Temporal Consciousness) grant from the ANR (Agence Nationale de la Recherche) and a CNRS Pluridisciplinary Exploratory Project in collaboration with the team of S. Meilhac (Institut Pasteur).
  
\appendix
\section{A variation formula for varifold metrics}
This appendix is dedicated to the proof of theorem \ref{theo:variation_formula}. We assume that $X$ is a smooth compact orientable submanifold of $E$ so that we can consider the canonical volume form $\sigma_{X}$ on $X$. Given a $C^{1}$ vector field $v$ on $E$ with compact support, we consider the $1$-parameter group of diffeomorphisms $\phi_{t}$ with $\phi_{0}=Id$ and $\partial_{t} |_{t=0} \phi_{t} =v$. The whole point of the theorem is to compute the variation of $\mu_{X}(\omega)$ when $X$ is deformed by vector field $v$, that is to say~: 
\begin{equation}
\label{eq:def_variationW}
\left. \dfrac{d}{dt} \right \vert_{t=0} \mu_{\phi_{t}(X)}(\omega)=\left. \dfrac{d}{dt} \right \vert_{t=0} \int_{X_t} \omega(x,T_{x}X_t) d \mathcal{H}^{d}(x)
\end{equation}
where $X_t=\phi_t(X)$.
Note that this generalizes the first variation of a varifold computed in \cite{Allard}, for which $\omega = \mathds{1}$. Here, $\omega$ is any $C^{1}$ function on $E \times G_{d}(E)$. As already explained in section \ref{sec:var}, we have $\mu_{\phi_{t}(X)}(\omega)=\mu_{X}(\phi_{t}^{\ast} \omega)$ and~:
\begin{equation}
\label{eq:variationW1}
 \left. \dfrac{d}{dt} \right \vert_{t=0} \mu_{\phi_{t}(X)}(\omega) = \int_{X} \left. \dfrac{d}{dt} \right \vert_{t=0} (\phi_{t}^{\ast} \omega)(x,T_{x}X) d \mathcal{H}^{d}(x) =\int_X\pounds_v\omega(x,T_xX)d \mathcal{H}^{d}(x)
\end{equation}
where $\pounds_v\omega$ stands for the Lie derivative of $\omega(x,V)$ along $v$. Hereafter, to simplify the notation, for a function $(x,V)\mapsto f(x,V)$, we will simply write $\int_Xf$ instead of $\int_Xf(x,T_xX)d \mathcal{H}^{d}(x)$.
Now, $(\phi_{t}^{\ast} \omega)(x,T_{x}X)=|d_{x}\phi_{t}.T_{x}X|\omega(\phi_{t}(x),d_{x}\phi_{t}(T_{x}X))$. Let us recall that $J_{t}=|d_{x}\phi_{t}.T_{x}X|$ is the volume change in the direction of the tangent space $T_{x}X$. Taking the derivative inside the integral at $t=0$ leads to three terms~: differentiate the function $\omega$ with respect to position, with respect to the tangent space direction and differentiate the volume change. The first term is the simplest one and an immediate computation shows that it equals $\left ( \dfrac{\partial \omega}{\partial x} | v \right )$. The two others are more involved. \\
\\
\textbf{Derivative of the volume change~:} For any vector field $u$ defined on $X$, we shall denote by $u^{\top}$ and $u^{\bot}$ the tangential and normal components of $u$ with respect to the tangent space of $X$ at each point. We also introduce the connection $\nabla_{\cdot}\cdotp$ on the ambient space and an orthonormal frame of tangent vector fields $(e_{i})_{i=1,..,d}$ on $X$. Now $J_{t}=\sqrt{\det(\langle d_{x}\phi_{t}(e_{i}),d_{x}\phi_{t}(e_{j}))_{i,j}}$ so a simple calculation shows that~:
\begin{equation*}
 \left. \dfrac{d}{dt} \right \vert_{t=0} J_{t} = \sum_{i=1}^{d} \langle e_{i}, \nabla_{e_{i}}v \rangle
\end{equation*}
Writing $v=v^{\top}+v^{\bot}$ provides a first term $\sum_{i=1}^{d} \langle e_{i}, \nabla_{e_{i}}v^{\top} \rangle$ which is the tangential divergence of the vector field $v^{\top}$ denoted usually $\mdiv_{X}(v^{\top})$. The second term becomes $\sum_{i=1}^{d} \langle e_{i}, \nabla_{e_{i}}v^{\bot} \rangle$. For all $i=1,..,d$, we have $\langle e_{i}, v^{\bot} \rangle = 0$ so that after differentiation we find that $\langle e_{i}, \nabla_{e_{i}}v^{\bot} \rangle = - \langle \nabla_{e_{i}} e_{i}, v^{\bot} \rangle$. Therefore~: 
\begin{eqnarray*}
 \sum_{i=1}^{d} \langle e_{i}, \nabla_{e_{i}}v^{\bot} \rangle &=& -\sum_{i=1}^{d} \langle \nabla_{e_{i}} e_{i}, v^{\bot} \rangle \\
 &=& -\left \langle \left (\sum_{i=1}^{d} \nabla_{e_{i}} e_{i} \right )^{\bot}, v^{\bot} \right \rangle
\end{eqnarray*}
In this last expression, we recognize the \textbf{mean curvature vector} to the submanifold $X$, which is the trace of the Weingarten map and is denoted $H_{X}$. As a result, we find that~:
\begin{equation*}
\int_{X} \omega  \left. \dfrac{d}{dt} \right \vert_{t=0} J_{t} = \int_{X} \omega  \mdiv_{X}(v^{\top}) - \int_{X} \omega \langle H_{X},v^{\bot} \rangle 
\end{equation*}
Now, we will show that the first term can be rewritten as a boundary integral. Indeed, if we denote by $\tilde{\omega}$ the function defined on $X$ by $\tilde{\omega}(x)=\omega(x,T_{x}X)$, we have $\mdiv_{X}(\tilde{\omega}v^{\top})=\tilde{\omega}\mdiv_{X}(v^{\top})+\langle \nabla \tilde{\omega} | v^{\top} \rangle$. Applying the divergence theorem on the orientable manifold $X$ gives~:
\begin{equation*}
 \int_{X} \omega  \mdiv_{X}(v^{\top}) = -\int_{X} \langle \nabla \tilde{\omega} | v^{\top} \rangle + \int_{\partial X} \omega\langle \nu , v^{\top} \rangle
\end{equation*}
where $\nu$ is the unit outward normal to the boundary. \\
\\
\textbf{Derivative with respect to tangent spaces~:} We now come to the last term in equation (\ref{eq:variationW1}), which is the variation of $\omega$ with respect to the tangent space in $G_{d}(E)$. As explained in the beginning of section \ref{sec:Grass}, we can identify the tangent space to $G_{d}(E)$ at $V$ with the space of linear applications $\mathcal{L}(V,V^{\bot})$. In addition, if we set $V_{t}=d_{x} \phi_{t}(T_{x}X)$, we have~:
\begin{equation*}
 \left. \dfrac{d}{dt} \right \vert_{t=0} V_{t} = p_{T_{x}X^{\bot}} \circ \nabla v \vert_{T_{x}X} \in \mathcal{L}(T_{x}X,(T_{x}X)^{\bot})
\end{equation*}
which we will note more concisely $\left. \dfrac{d}{dt} \right \vert_{t=0} V_{t} = \nabla^{\bot}v$. The space $\mathcal{L}(T_{x}X,(T_{x}X)^{\bot})$ being trivially isomorphic to $T_{x}X^{\bot} \otimes (T_{x}X)^{*}$, we can consider $\nabla^\bot v$ being an element of that space and we introduce $\dfrac{\partial \omega}{\partial V}$ which we therefore identify to an element of $(T_{x}X^{\bot})^{*} \otimes T_{x}X$ : $\dfrac{\partial \omega}{\partial V} = \sum_{j=d+1}^{n} \eta_{j}^{*}\otimes \alpha_{j}$ for $(\eta_{d+1},..,\eta_{n})$ an orthonormal frame of $T_{x}X^{\bot}$ and $(\alpha_{j})$ vectors of $T_{x}X$ (as usual $\eta^*$ denotes the linear form $\langle \eta,.\rangle$). Then the variation we wish to compute is~:
\begin{equation*}
  \left ( \dfrac{\partial \omega}{\partial V} | \nabla^\bot v \right )=\left ( \dfrac{\partial \omega}{\partial V} | \nabla v \right ) = \sum_{j=d+1}^{n} \langle \eta_{j}, \nabla_{\alpha_{j}}v \rangle
\end{equation*}
If we introduce $\left( \dfrac{\partial \omega}{\partial V} | v \right )= \sum_{j=d+1}^{n} \eta_{j}^{*}(v)\alpha_{j} = \sum_{j=d+1}^{n} \langle \eta_{j},v \rangle\alpha_{j}$ which is a tangent vector field on $X$, we have~:
\begin{equation*}
 \mdiv_{X}\left( \dfrac{\partial \omega}{\partial V} | v \right ) = \sum_{i=1}^{d} \sum_{j=d+1}^{n} \left( \langle e_{i}, \nabla_{e_{i}} \alpha_{j} \rangle\langle \eta_{j},v \rangle + \langle e_{i}, \langle \nabla_{e_{i}}\eta_{j},v \rangle \alpha_{j} \rangle + \langle e_{i}, \langle \eta_{j},\nabla_{e_i}v \rangle \alpha_{j} \rangle \right)
\end{equation*}
The last term in the sum is also $\sum_{j=d+1}^{n} \langle \eta_{j},\nabla_{\alpha_{j}}v \rangle$, which is nothing else than $\left ( \dfrac{\partial \omega}{\partial V} | \nabla v \right )$. As for the two other terms in the sum, it's easy to see that it equals~:
\begin{equation*}
 \left ( \sum_{i=1}^{d} \langle e_{i}, \nabla_{e_{i}}\sum_{j=d+1}^{n} \eta_{j}^{*}\otimes \alpha_{j} \rangle | v \right ) = ( \mdiv_{X} \left (\dfrac{\partial \omega}{\partial V} \right ) | v )
\end{equation*}
So get eventually that~: 
\begin{equation}
\label{eq:variationW2}
 \left ( \dfrac{\partial \omega}{\partial V} | \nabla v \right ) = \mdiv_{X}\left( \dfrac{\partial \omega}{\partial V} | v \right ) - ( \mdiv_{X} \left (\dfrac{\partial \omega}{\partial V} \right ) | v )
\end{equation}
Integrating equation (\ref{eq:variationW2}) over the submanifold $X$ and using as previously the divergence theorem, we find that~:
\begin{equation}
\label{eq:variationW3}
 \int_{X} \left ( \dfrac{\partial \omega}{\partial V} | \nabla v \right ) = \int_{\partial X} \langle \nu , \left( \dfrac{\partial \omega}{\partial V} | v \right ) \rangle - \int_{X} ( \mdiv_{X} \left (\dfrac{\partial \omega}{\partial V} \right ) | v ) 
\end{equation}

\textbf{Synthesis~:} Grouping all the different terms obtained so far, we get the following ~:
\begin{eqnarray*}
 \int_X \pounds_v\omega &=& \int_{X} \left ( \dfrac{\partial \omega}{\partial x} - \mdiv_{X} \left (\dfrac{\partial \omega}{\partial V} \right ) | v \right ) -\int_{X} \langle \nabla \tilde{\omega} | v^{\top} \rangle - \omega\langle H_{X} | v^{\bot} \rangle \\
 & & + \int_{\partial X} \langle \nu, \left( \dfrac{\partial \omega}{\partial V} | v \right ) + \omega v^{\top} \rangle
\end{eqnarray*}
We remind that $\tilde{\omega}(x)=\omega(x,T_{x}X)$ so $(\nabla \tilde{\omega} | v^{\top} ) = \left ( \dfrac{\partial \omega}{\partial x} | v^{\top} \right ) + \left ( \dfrac{\partial \omega}{\partial V} | \nabla v^{\top} \right )$ and applying the result of equation (\ref{eq:variationW2}) to $v^{\top}$~:
\begin{equation*}
\left ( \dfrac{\partial \omega}{\partial V} | \nabla v^{\top} \right ) = \mdiv_{X}\left( \dfrac{\partial \omega}{\partial V} | v^{\top} \right ) - ( \mdiv_{X} \left (\dfrac{\partial \omega}{\partial V} \right ) | v^{\top} )
\end{equation*}
Using the divergence theorem on $\mdiv_{X}\left( \dfrac{\partial \omega}{\partial V} | v^{\top} \right )$ and the equality $v=v^\top+v^\bot$ we find eventually that~:
\begin{equation*}
\int_X \pounds_v\omega = \int_{X} \left ( \dfrac{\partial \omega}{\partial x} - \mdiv_{X} \left (\dfrac{\partial \omega}{\partial V} \right ) - \omega H_{X} | v^{\bot} \right ) + \int_{\partial X} \langle \nu, \left( \dfrac{\partial \omega}{\partial V} | v \right ) + \omega v^{\top} \rangle
\end{equation*}
which proves the result of theorem \ref{theo:variation_formula}. 

\section{Discrete computations for varifold LDDMM}
As explained in section \ref{sec:var_LDDMM_algo}, the LDDMM algorithm relies on the computation of a distance between the deformed source object $\phi(\mathcal{O}_{1})$ and the target object $\mathcal{O}_{2}$, and on the gradient of this distance with respect to final points $q_{1}$ of the object $\mathcal{O}_{1}$. In our context of unoriented shapes modeled as varifolds, we use distances provided by kernels and their associated RKHS presented in section \ref{sec:var_kernel}. Those are the tensor product of a kernel $k_{e}$ on the space $\mathbb{R}^3$ and a kernel $k_{t}$ on the Grassmann manifold. If $W$ is the RKHS corresponding to $k=k_{e}\otimes k_{t}$, the attachment distance we have is then~:
 \begin{equation}
\label{eq:attach_general}
 d(\phi_{1}^{v}.\mathcal{O}_{1},\mathcal{O}_{2})= \|(\phi_{1}^{v})_{\ast} \mu_{\mathcal{O}_{1}} - \mu_{\mathcal{O}_{2}} \|_{W^{\ast}}^{2} = \|\mu_{\phi_{1}^{v}(\mathcal{O}_{1})} - \mu_{\mathcal{O}_{2}} \|_{W^{\ast}}^{2} 
 \end{equation} 
Since $\phi_{1}^{v}(\mathcal{O}_{1})$ is the transported source object by the vector field at time $1$ (represented by the set of points $q_{1}$ with the previous notations), formally the problem reduces to express explicitly quantities like $\|\mu_{X} - \mu_{Y} \|_{W^{\ast}}^{2}$ where $X$ and $Y$ are two shapes of the same dimension and compute the gradient of such terms with respect to the points of $X$.  \\
\\ 
We assume now that $X$ and $Y$ are two unoriented $d$-dimensional shapes (possibly disconnected) given respectively as sets of vertex $(x_{k})_{k=1,..,N}$, $(y_{l})_{l=1,..,M}$ together with sets of unoriented simplexes $(f_{i}^{1},..,f_{i}^{d})_{i=1,..,n}$ and $(g_{j}^{1},..,g_{j}^{d})_{j=1,..,m}$. The computation of the varifold representations $\mu_{X}$ and $\mu_{Y}$ of each shape was explained in section \ref{sec:Grass} and provided by equations (\ref{eq:discrete_curve_varifold}) for curves, (\ref{eq:discrete_surface_varifold}) for surfaces. We write $\mu_{X}=\sum_{i=1}^{n} l_{i}.\delta_{(p_{i},U_{i})}$ and $\mu_{Y}=\sum_{j=1}^{m} \lambda_{j}.\delta_{(q_{j},V_{j})}$. Then the attachment distance becomes~:
 \begin{eqnarray}
\label{eq:attach_curve}
 A &=& \|\mu_{X}\|_{W^{\ast}}^{2}-2\langle \mu_{X}, \mu_{Y} \rangle_{W^{\ast}}+\|\mu_{Y}\|_{W^{\ast}}^{2} \nonumber \\
&=& \sum_{i,j=1..n} l_{i} l_{j} k_{e}(p_{i},p_{j}).k_{t}(U_{i},U_{j}) \nonumber \\
& &-2 \sum_{i=1..n}\sum_{j=1..m} l_{i} \lambda_{j} k_{e}(p_{i},q_{j}).k_{t}(U_{i},V_{j}) \nonumber \\ 
& &+ \sum_{i,j=1..n} \lambda_{i} \lambda_{j} k_{e}(q_{i},q_{j}).k_{t}(V_{i},V_{j})
 \end{eqnarray}  
Thus, computing the distance between $X$ and $Y$ consists basically in computing the varifold representation of the shapes from their sets of points and meshes and then make repeated kernel evaluations. If $k_{e}$ and $k_{t}$ are specified, the previous expression becomes totally explicit. \\
The more technical part is the computation of the gradient of the previous distance with respect to the $(x_{k})$'s, the points of the first shape. In equation (\ref{eq:attach_curve}), the attachment is a function of the $p_{i}$, $l_{i}$ and $U_{i}$ but all these terms are themselves functions of the $(x_{k})$'s, respectively as the centers, lengths and tangent directions to every cell of $X$. The resulting function that we have denoted $g((x_{k}))$ can be therefore differentiated as a composition~:
\begin{equation}
\label{eq:grad1}
 \partial_{x_{k}}g= \sum_{i=1}^{n} \left ( \partial_{p_{i}}A \circ \partial_{x_{k}}p_{i} + \partial_{l_{i}}A \circ \partial_{x_{k}}l_{i} + \partial_{U_{i}}A \circ \partial_{x_{k}}U_{i} \right )
 \end{equation}    
The differentials of $A$ with respect to the $p_{i}$, $l_{i}$ and $U_{i}$ are easily expressed from equation (\ref{eq:attach_curve}), involving the differentials of the kernels $k_{e}$ and $k_{t}$.
\begin{equation}
\label{eq:grad2}
\left\{
  \begin{array}[h]{l}
 \partial_{p_{i}} A = \sum_{j=1}^{n} l_{i} l_{j} [\partial_{1}k_{e}(p_{i},p_{j})+\partial_{2}k_{e}(p_{j},p_{i})].k_{t}(U_{i},U_{j})\\
 \hspace{1cm} -2 \sum_{j=1}^{m} l_{i} \lambda_{j} \partial_{1}k_{e}(p_{i},q_{j}).k_{t}(U_{i},V_{j}) \\
\\
 \partial_{U_{i}} A = \sum_{j=1}^{n} l_{i} l_{j} k_{e}(p_{i},p_{j}).[\partial_{1}k_{t}(U_{i},U_{j})+\partial_{2}k_{e}(U_{j},U_{i})]\\
 \hspace{1cm} -2 \sum_{j=1}^{m} l_{i} \lambda_{j} k_{e}(p_{i},q_{j}).\partial_{1}k_{t}(U_{i},V_{j}) \\
\\
 \partial_{l_{i}} A = 2\sum_{j=1}^{n} l_{j} k_{e}(p_{i},p_{j}).k_{t}(U_{i},U_{j}) \\
\hspace{1cm} -2\sum_{j=1}^{m} \lambda_{j} k_{e}(p_{i},q_{j}).k_{t}(U_{i},V_{j})
  \end{array}\right.
\end{equation}
where $\partial_{1}$ and $\partial_{2}$ denotes the derivative with respect to the first and second argument. We remind (cf section \ref{sec:kernel_Grassmann}) that, in our approach, kernel $k_{t}$ is the restriction of a kernel defined on the vector space $\mathcal{L}(E)$ so that writing its derivatives is natural. \\ 
Now, the last step consists in computing derivatives of $p_{i}$, $l_{i}$ and $U_{i}$ with respect to the points of $X$. These depend on the dimension and codimension of the shape. In the case of curves and surfaces embedded in $\mathbb{R}^{3}$, it can be simply done by differentiating equations (\ref{eq:discrete_curve_varifold}) for curves or equations (\ref{eq:discrete_surface_varifold}) for surfaces. We give the details for these two cases below. \\
In the case of \textbf{curves}, we have $p_{i}=\frac{(x_{f_{i}^{1}}+x_{f_{i}^{2}})}{2}$, $l_{i}=|x_{f_{i}^{2}}-x_{f_{i}^{1}}|$. The tangent space direction can be represented as a single normalized vector $u_{i}=\frac{x_{f_{i}^{2}}-x_{f_{i}^{1}}}{|x_{f_{i}^{2}}-x_{f_{i}^{1}}|}$, which gives for all $i \in \{1,..,n\}$ and $k \in \{1,..,N\}$
\begin{equation}
\label{eq:grad_var_curve}
\left\{
  \begin{array}[h]{l}
 \partial_{x_{k,s}} p_{i} = \dfrac{1}{2}(\delta_{\{k=f_{i}^{1}\}}+\delta_{\{k=f_{i}^{2}\}})  \\
\\
 \partial_{x_{k,s}} u_{i} = \dfrac{1}{|x_{f_{i}^{2}}-x_{f_{i}^{1}}|} \left (e_{s}-\dfrac{x_{f_{i}^{2},s}-x_{f_{i}^{1},s}}{|x_{f_{i}^{2}}-x_{f_{i}^{1}}|}u_{i} \right ) \delta_{\{k=f_{i}^{2}\}}\\
  \hspace{1cm} -\dfrac{1}{|x_{f_{i}^{2}}-x_{f_{i}^{1}}|} \left (e_{s}-\dfrac{x_{f_{i}^{2},s}-x_{f_{i}^{1},s}}{|x_{f_{i}^{2}}-x_{f_{i}^{1}}|} u_{i} \right ) \delta_{\{k=f_{i}^{1}\}}\\
\\
 \partial_{x_{k,s}} l_{i} = \dfrac{x_{f_{i}^{2},s}-x_{f_{i}^{1},s}}{|x_{f_{i}^{2}}-x_{f_{i}^{1}}|} \delta_{\{k=f_{i}^{2}\}}-\dfrac{x_{f_{i}^{2},s}-x_{f_{i}^{1},s}}{|x_{f_{i}^{2}}-x_{f_{i}^{1}}|} \delta_{\{k=f_{i}^{1}\}}
  \end{array}\right.
\end{equation}
In the previous equations, $x_{k,s}$ denotes the number $s$ coordinate of $x_{k}$ in the embedding space's canonical basis $(e_{s})$. These equations can be simply interpreted and implemented as the way to distribute the gradient computed with respect to the segments over the points of the shape $X$. \\  
In the case of \textbf{triangulated surfaces} in $\mathbb{R}^3$, we have $p_{i}=\frac{(x_{f_{i}^{1}}+x_{f_{i}^{2}}+x_{f_{i}^{3}})}{3}$, \\
$l_{i}=|(x_{f_{i}^{2}}-x_{f_{i}^{1}})\wedge(x_{f_{i}^{3}}-x_{f_{i}^{1}})|$ and again the tangent space is simply encoded by the normalized normal vector $u_{i}=\frac{(x_{f_{i}^{1}}-x_{f_{i}^{2}})\wedge(x_{f_{i}^{2}}-x_{f_{i}^{3}})}{|(x_{f_{i}^{1}}-x_{f_{i}^{2}})\wedge(x_{f_{i}^{2}}-x_{f_{i}^{3}})|}$ so that for all $i \in \{1,..,n\}$ and $k \in \{1,..,N\}$
\begin{equation}
\label{eq:grad_var_surface}
\left\{
  \begin{array}[h]{l}
 \partial_{x_{k,s}} p_{i} = \dfrac{1}{3}(\delta_{\{k=f_{i}^{1}\}}+\delta_{\{k=f_{i}^{2}\}}+\delta_{\{k=f_{i}^{3}\}})  \\
\\
 \partial_{x_{k,s}} u_{i} = \dfrac{1}{l_{i}} \left ( e_{s}\wedge (x_{f_{i}^{2}}-x_{f_{i}^{3}})-\left \langle e_{s}\wedge (x_{f_{i}^{2}}-x_{f_{i}^{3}}) \ , \ u_{i} \right \rangle u_{i}  \right ) \delta_{\{k=f_{i}^{1}\}} \\
 \hspace{2cm} + ...\\
\\
 \partial_{x_{k,s}} l_{i} = \left \langle e_{s}\wedge (x_{f_{i}^{2}}-x_{f_{i}^{3}}) \ , \ u_{i} \right \rangle \delta_{\{k=f_{i}^{1}\}} + ...
  \end{array}\right.
\end{equation}
We wrote ``$+...$`` to mean that similar terms are obtained for $k=f_{i}^{1}$ and $k=f_{i}^{2}$. \\
Combining the last relations with equations (\ref{eq:grad1}) and (\ref{eq:grad2}) provides a full description of the gradient computation in practice.

\tableofcontents

\bibliographystyle{abbrv}
\bibliography{biblio}

\end{document}